\documentclass[12pt,twoside]{article}
\usepackage{fix-cm}

\usepackage[left=1.15in,right=1.15in,top=1.22in,bottom=1.22in]{geometry}
\usepackage{setspace}
\usepackage{caption}
\usepackage{graphicx}
\graphicspath{{figures/}}
\usepackage{amsmath, amssymb, amsthm}
\usepackage{booktabs}
\usepackage{threeparttable}
\usepackage{array}
\usepackage{mathtools}
\usepackage{times}

\usepackage{natbib}
\setcitestyle{authoryear,round,semicolon}

\usepackage[dvipsnames]{xcolor}
\usepackage[hyperfootnotes=false]{hyperref}
\definecolor{linkblue}{RGB}{0,65,130}
\hypersetup{
  colorlinks=true,
  linkcolor=linkblue,
  citecolor=linkblue,
  urlcolor=linkblue,
  pdftitle={Learning about Treatment Effects in Panels under Unknown Interference},
  pdfauthor={Shengbin Wei},
  pdfsubject={Identification and inference for panel treatment effects under unknown interference},
  pdfkeywords={panel treatment effects, spillovers, interference, partial identification, exposure mappings, test inversion}
}
\usepackage{enumitem}
\usepackage{dsfont}
\usepackage{fancyhdr}
\providecommand{\mathbbm}[1]{\mathbf{#1}}

\usepackage{titlesec}
\titleformat{\section}
  {\centering\normalfont\large\scshape}
  {\thesection.}{0.75em}{}
\titlespacing*{\section}{0pt}{2.6ex plus 0.8ex minus 0.2ex}{1.4ex plus 0.2ex}
\titleformat{\subsection}
  {\centering\normalfont\normalsize\scshape}
  {\thesubsection.}{0.6em}{}
\titlespacing*{\subsection}{0pt}{2.0ex plus 0.6ex minus 0.2ex}{0.8ex plus 0.2ex}
\titleformat{\subsubsection}
  {\normalfont\normalsize\itshape}
  {\thesubsubsection.}{0.6em}{}
\titlespacing*{\subsubsection}{0pt}{1.6ex plus 0.4ex minus 0.2ex}{0.6ex plus 0.2ex}

\usepackage{titling}
\pretitle{\begin{center}\Large\scshape}
\posttitle{\par\end{center}\vskip 1em}
\preauthor{\begin{center}\Large\scshape}
\postauthor{\par\end{center}\vskip 0.25em}
\predate{\begin{center}\normalsize}
\postdate{\par\end{center}\vskip 0.25em}

\renewenvironment{abstract}
  {\begin{quote}\small}
  {\end{quote}\vspace{0.5em}}

\fancypagestyle{plain}{%
  \fancyhf{}%
  \fancyhead[LE]{\raisebox{-0.07in}[0pt][0pt]{\thepage}}%
  \fancyhead[RO]{\raisebox{-0.07in}[0pt][0pt]{\thepage}}%
}

\newtheorem{theorem}{Theorem}[section]
\newtheorem{proposition}{Proposition}[section]
\newtheorem{lemma}{Lemma}[section]
\newtheorem{corollary}{Corollary}[section]
\newtheorem{assumption}{Assumption}[section]
\theoremstyle{remark}
\newtheorem{remark}{Remark}[section]
\theoremstyle{definition}
\newtheorem{definition}{Definition}[section]
\newtheorem{example}{Example}[section]
\newtheorem{procedure}{Procedure}[section]
\newcommand{\PP}{\mathbb{P}}

\numberwithin{equation}{section}

\title{\MakeUppercase{Learning about Treatment Effects in Panels under Unknown Interference}%
\thanks{I am grateful to David Hughes, Shakeeb Khan, Arthur Lewbel, and Zhijie Xiao for their
mentorship. I also benefited from comments by Harold Chiang, Zequn Jin, Yechan
Park, Andreas Petrou-Zeniou, and Liang Zhong, and from discussions with
participants at NY Camp Econometrics 2026, NEEPC 2026, the IAER Econometrics
Workshop 2026, and CCER SI 26. All remaining errors are mine.}}
\author{Shengbin Wei\thanks{Department of Economics, Boston College, Maloney Hall, 140 Commonwealth
Avenue, Chestnut Hill, MA 02467-3859. Email:
\href{mailto:shengbin.wei@bc.edu}{shengbin.wei@bc.edu}.}}
\date{August 13, 2026\\
\vspace{0.35em}
\textit{Working Paper}}

\begin{document}
\maketitle

\begin{abstract}
\setlength{\emergencystretch}{4em}
\sloppy
When comparison units may also respond to treatment, panel comparisons reflect
both the treatment effect and spillovers. If the interference pattern is unknown,
observed outcomes alone do not separate the two. I characterize what can
nevertheless be learned from panel outcomes under general restrictions,
without requiring an exposure mapping or prior classification of affected
donors. The framework scales validity bounds for every convex donor weight by
its fit before treatment and combines these bounds with prespecified
restrictions tailored to the application. The validity bounds constrain the
treatment effect relative to spillovers, while the additional restrictions
determine its possible values. Together these restrictions yield a sharp
identified set. When the additional restrictions have a finite linear
representation, checking whether a proposed treatment effect is compatible
with the model reduces exactly to asking whether a finite linear system has a
solution. Bootstrap calibration tests this condition. Inverting these tests
uniformly controls, in large samples, the probability of falsely excluding
each compatible value. In an application to the Legal Arizona Workers Act,
the resulting 95 percent inversion sets contain effects of both signs across
all reported specifications, leaving the sign of the treatment effect
unresolved.
\fussy
\end{abstract}

\section{Introduction}

Panel treatment effect methods use untreated units to construct counterfactual
outcomes. Their interpretation changes when the policy also affects the
comparison pool. The Legal Arizona Workers Act (LAWA) provides a useful
example. If Arizona's employment verification law deterred migrants from
Arizona or redirected them to other states, outcomes in comparison states may
also respond. Comparisons between Arizona and those states then reflect both
the treatment effect and spillovers. The same concern arises when policies
redirect activity across jurisdictions, providers, or locations or propagate
through social and economic networks, as in tax, environmental, health, and
place-based settings
\citep{agrawal2015taxgradient,lepissier2021climate,
alexander2023hospitalclosures,fischer2024obstetric,
banerjee2013diffusion,cai2015social}.

Applications often provide substantive restrictions on treatment and spillover
effects even when the spillover pattern remains unknown. Such restrictions
may come from outcome support, aggregate accounting, signs, orderings,
magnitude information, or a structural exposure model. I treat the full
vector of spillovers to comparison units as unknown and impose restrictions
directly on it, without requiring an exposure mapping or prior classification
of affected donors. Recent work formalizes the basic identification problem in
a difference-in-differences setting with two groups: under a parallel-trends
condition for no-policy outcomes and otherwise unknown interference, the
canonical estimand identifies the average total effect on treated units only
relative to the average spillover on controls
\citep{mealli2026difference}. This paper develops the corresponding problem
with many donors. Additional weights are not additional data.
They become informative when the same treatment effect and spillover vector
must satisfy a common validity rule in several comparison directions. Building
on \citet{liu2025synthetic}, I impose this common rule over the full simplex
generated by a prespecified donor pool.

The paper contributes a sharp identification result, an exact finite
representation of the restrictions over all donor weights, and a uniform
candidatewise inference result for the structured systems generated by that
representation. These three pieces separate the economic content of the model
from computation and sampling theory.

For each weight, the rule bounds its latent no-policy gap after treatment by a
prespecified multiple of its typical absolute discrepancy before treatment.
The distinctive
commitments are that the rule holds uniformly over the prespecified donor
simplex and that its baseline allowance has a zero floor when population
discrepancy before treatment vanishes. A condition on the span of factor
movements provides a sufficient condition for the full system. To make this
multiplier easier to interpret, I construct factor placebo indices by
leaving out one period at a time. These placebos put the multiplier on an
observed scale, following the calibration principle of
\citet{hsu2013calibrating}, but they benchmark rather than estimate the latent
threshold after treatment.

The analysis considers one candidate value of the treatment effect at a time.
A candidate is compatible if some unobserved no-policy outcomes and spillovers
satisfy all maintained restrictions. Comparison validity constrains the
treatment effect relative to spillovers, while restrictions tailored to the
application determine which levels of the treatment effect remain possible.
Together, they characterize the sharp identified set and show whether it has
an upper endpoint, a lower endpoint, or both.

Under the baseline finite representation, a candidate is compatible exactly
when a linear system of fixed dimension has a solution. An exact equivalence
replaces the continuum of comparison weights with finitely many inequalities
and shows when interior weights add information beyond individual donors.
Farkas' alternative provides a finite certificate of incompatibility. Building
on the normalized solvability statistic and bootstrap calibration of
\citet{goff2025inference}, I establish uniform validity for this model-generated
family of candidate systems and invert the resulting tests over candidate
values. The analysis exploits their common primitive sampling perturbation and
uses regularity conditions tailored to compatible right-hand sides and
structurally reachable certificate directions. In large samples, the procedure
uniformly controls the probability of falsely excluding each compatible
candidate, while candidates that remain a fixed normalized distance from
compatibility in the population are consistently excluded. These results are
conditional on the donor pool, chosen validity envelope, admissible rule, and
its finite representation.

I apply the method to the LAWA application of \citet{bohn2014lawa}. The
application combines state-level outcomes with outcome support and a gross
spillover budget scaled by population. The main specification bounds total
absolute responses across comparison units by twice the counterfactual
population of Arizona's target group; I also report bounds equal to one and
four times that population. This formulation allows interstate redistribution, changes in
the national stock, and responses outside the comparison pool without
requiring an exact relocation identity. Across the reported validity envelopes and
all three budget calibrations, neither zero nor the decline of 1.50 percentage
points reported by \citet{bohn2014lawa} is rejected. The original magnitude
therefore remains compatible with the data and restrictions, but the analysis
does not determine the sign of the effect after accounting for sampling
uncertainty.

The Monte Carlo design separates identification geometry from sampling
performance. When pre-treatment discrepancies are sign aligned, the
full simplex and donor vertices yield the same population identified set. When
convex aggregation produces near-exact fit, interior weights reduce the
width of that set by 86 percent. Across both geometries, false exclusion at
compatible endpoints shows no evidence of excessive rejection in the simulated
cells, while power at fixed distances rises with sampling precision.

The remainder of the paper proceeds as follows.
\hyperref[sec:setting]{Section~\ref*{sec:setting}} sets up the model,
\hyperref[sec:identification]{Section~\ref*{sec:identification}} develops the
exact finite representation of candidate compatibility and its projection
consequences, and
\hyperref[sec:inference]{Section~\ref*{sec:inference}} constructs compatibility
tests for fixed candidate values and their inversion.
\hyperref[sec:empirical_illustration]{Section~\ref*{sec:empirical_illustration}}
revisits LAWA, and \hyperref[sec:monte_carlo]{Section~\ref*{sec:monte_carlo}}
studies performance in finite samples. The appendix collects proofs,
implementation details, and additional results.

\subsection{Related Literature}\label{sec:related_literature}

The paper is closest to panel methods for treatment effects that allow an
intervention to contaminate comparison units. Some synthetic control
approaches specify a linear, spatial, or exposure structure, or identify which
units are affected, and
jointly recover direct and spillover effects
\citep{cao2019estimation,distefano2024inclusive,grossi2025direct,
sakaguchi2026bayesian}. \citet{melnychuk2024synthetic} compares strategies for
constructing counterfactuals for the treated unit when donors may be contaminated.
Others downweight, screen, or
exclude donors considered likely to be affected
\citep{oriordan2025spillover,fernandezmorales2026bayesian}. I instead treat the
full vector of spillovers to comparison units as unknown. Joint validity
across contaminated comparisons and prespecified restrictions on that vector
partially identify the treatment effect. This route requires neither a
maintained exposure mapping, a pure donor subset, nor point identification of
the spillover profile. \citet{celli2026identifying}
contrast counterfactuals based on controls and forecasts under pervasive
interference; the framework here retains panel comparisons based on controls but
replaces point identification with explicit compatibility restrictions.

The broader interference literature defines direct and spillover effects using
treatment assignment or exposure mappings and derives identification under
experimental, network, or neighborhood structures
\citep{manski2013identification,sobel2006randomized,hudgens2008toward,
aronow2017general,vazquezbare2023spillover,leung2020treatment,
leung2022causal,forastiere2021identification,huber2021framework}. Under
randomized assignment, \citet{savje2021average} define an average treatment
effect that remains meaningful under limited but otherwise unknown
interference. Extensions to
panel data adapt parallel trends restrictions to spatial or neighborhood
interference \citep{butts2021difference,xu2023difference}.
\citet{mealli2026difference} show that, under a parallel-trends condition for
no-policy outcomes and otherwise unknown interference, the canonical
difference-in-differences estimand identifies the average total effect on
treated units net of the average spillover on controls. The present paper starts from
the analogous accounting relation with many donors and derives sharp sets from
validity restrictions indexed by weights and joint restrictions on the spillover
vector. Work on misspecified exposure mappings is complementary: it studies
robustness to the mapping itself, whereas the admissible set here directly restricts the
resulting spillover profile and can incorporate implications of a maintained
exposure model \citep{savje2024causal,schroder2026causal}.

Synthetic control and related panel methods use pre-treatment outcomes to form
counterfactual comparisons
\citep{abadie2003economic,abadie2010synthetic,abadie2015comparative};
\citet{abadie2021using} surveys the design and its extensions. Like synthetic
control and synthetic difference-in-differences
\citep{arkhangelsky2021synthetic}, the framework here uses pre-treatment
outcomes to constrain post-treatment comparisons. Its distinction is that a
sensitivity envelope is imposed over every weight in a prespecified donor
simplex rather than attached to one estimated comparison. In all three
approaches, the connection between pre-treatment fit and post-treatment
validity ultimately rests on maintained restrictions on untreated outcomes.
Interactive factor models provide one common way to extrapolate from outcomes
before treatment to outcomes after treatment
\citep{bai2009panel,gobillon2016regional,xu2017generalized}.
Analyses of imperfect fit characterize its consequences, propose bias
corrections, or develop alternative synthetic control estimators
\citep{ferman2021synthetic,benmichael2021augmented,powell2026imperfect}.
Here the observed discrepancy instead enters directly as the scale of a
maintained validity envelope.

The closest connection is \citet{liu2025synthetic}, who studies the identifying
content of multiple comparison weights. In that paper, each weight that exactly
balances trends before treatment indexes a
possible counterfactual, and the identified set ranges over those alternatives.
Here the restrictions indexed by weights are intersected: every convex weight must
satisfy its own fit-scaled restriction for the same treatment effect and
spillover vector. Expanding the comparison domain therefore adds restrictions
to one joint feasible set. \citet{callaway2025beyond} likewise use sets of
close comparison groups, but target robustness across alternative panel
identification strategies. The fit scale is related to the relative-magnitude
sensitivity analysis of \citet{rambachan2023more}, which bounds post-treatment
departures using pre-treatment behavior. \citet{ferguson2021assessing}
calibrate misspecification in synthetic control using observable placebo
prediction errors. The calibration here instead benchmarks a common validity
envelope over all donor weights and combines it with restrictions on
spillovers to comparison units. It also adapts the calibration principle of
\citet{hsu2013calibrating}. Here, indices constructed by leaving out one period
at a time put specifications with fixed envelope values on an observed scale
without treating pre-treatment movements as estimates of the missing
post-treatment gap.

Finally, the paper relates to inference under partial identification and to
inference for optimization and problems involving linear systems. General
methods cover identified sets, moment inequalities, and low-dimensional projections
\citep{chernozhukov2007estimation,rosen2008confidence,
andrews2010inference,romano2010inference,romano2014practical,
chernozhukov2013intersection,bugni2017inference,kaido2019confidence,
gafarov2025simple}. Related work studies directionally differentiable maps,
linear moment models, and linear systems with known or estimated coefficients
\citep{fang2019inference,cho2023simple,andrews2023linear,fang2023large,
goff2025inference,bai2026linear}.
\citet{cox2025testing} develop tests for inequalities that are linear in
nuisance parameters, including applications in which target parameters are
bounded by linear programs.
Goff and Mbakop develop general tests for whether a linear system has a
solution. The inference result here is a structured specialization of their
procedure. I show that candidate compatibility in this model has an exact
finite representation with estimated coefficients, adapt their calibration to
the resulting candidate-indexed family, and establish validity under
model-tailored conditions on a uniform feasibility error bound and certificate
variance
before inverting the tests over candidate values.

\section{Setup and Notation}\label{sec:setting}
\noindent
I focus on settings where a binary treatment is implemented at an aggregate
level, such as states or countries. Denote each aggregate unit by
\(k\in\{1,\ldots,K\}\), where \(K\) is the total number of units, and let
\(k=1\) be the treated unit. Pre-treatment periods are indexed by
\(t=1,\ldots,T_0\), where \(T_0\) is the last pre-treatment period; \(T\)
labels one scalar post-treatment target rather than an additional date in this
sequence. In the LAWA application,
unit \(1\) is Arizona; Section~\ref{sec:empirical_illustration} gives the
corresponding sample, timing, and donor pool details.
Assume throughout that \(K\ge2\) and \(T_0\ge2\).
Throughout, \emph{prespecified} means fixed as part of the maintained
specification before estimating the corresponding candidate system and
evaluating its inversion; it does not denote formal preregistration.
I use one treated aggregate unit and one scalar post-treatment target for
notational simplicity. The accounting relation and finite-input construction
also apply when this scalar target is a prespecified linear aggregate of several
treated units or post-treatment outcomes. A simultaneous vector of
treatment-effect targets would additionally require definitions of the corresponding
spillover objects, admissible restrictions, and covariance input and is not
developed in the present notation.\footnote{The treated unit, comparison pool,
and treatment timing are taken as given; see
\citet{imbens2023identification} for an alternative framework in which aspects
of the design are viewed as random.}

For identification, each aggregate outcome is a population mean. Write
\(\mu_t^k\) for the observed mean of unit \(k\) in period \(t\). Let
\(\mathbf d^0=(0,\ldots,0)\) denote the baseline assignment without the focal
intervention and let
\(\mathbf d^1=(1,0,\ldots,0)\) denote the realized assignment in which unit
\(1\) receives the policy and the comparison units remain directly untreated.
Write \(\mu_t^k(\mathbf d)\) for the potential mean of unit \(k\) under
assignment \(\mathbf d\). In the LAWA example,
\(\mu_t^k(\mathbf d^0)\) is the population share in state \(k\) that would prevail
absent LAWA, while \(\mu_T^k(\mathbf d^1)\) is the share under the realized
LAWA assignment. The contrast changes only the focal intervention; background
policies remain part of the potential-outcome environment. Accordingly,
``no-policy'' below means without the focal intervention. I impose consistency
and no anticipation with respect to that intervention: for every \(k\),
\[
\mu_t^k=\mu_t^k(\mathbf d^0)=\mu_t^k(\mathbf d^1)\quad (t\le T_0),
\qquad
\mu_T^k=\mu_T^k(\mathbf d^1).
\]

The treatment and spillover effects compare these two complete assignments.
The target is the treatment effect on the treated aggregate unit,
\begin{equation}\label{eq:treatment}
\tau := \mu_T^1(\mathbf d^1)-\mu_T^1(\mathbf d^0),
\end{equation}
which changes unit \(1\)'s assignment while holding every comparison unit's
assignment at zero. The effect of the same assignment contrast on a
comparison unit is its spillover effect,
\begin{equation}\label{eq:spillover}
s_k := \mu_T^k(\mathbf d^1)-\mu_T^k(\mathbf d^0),
\qquad k=2,\ldots,K.
\end{equation}
Write \(s=(s_2,\ldots,s_K)^\top\). The benchmark \(s=0\) corresponds to zero
spillovers to comparison units. The analysis below treats \(s\) as unknown and
asks what contaminated comparisons and restrictions on \(s\) reveal about
\(\tau\).

The potential means are reduced-form equilibrium outcomes under each complete
assignment. Thus \(\tau\) and \(s\) can incorporate geographic propagation
through relocation, displacement, commuting, disease transmission, or regional
market adjustment, as well as network propagation through social, trade, or
input-output links \citep[e.g.,][]{miguel2004worms,greenstone2010agglomeration,
banerjee2013diffusion,cai2015social}. These channels may be nonlinear or
reciprocal; \(s_k\) records their net effect on comparison unit \(k\).
Throughout the paper, the treatment effect refers to the assignment contrast in
\eqref{eq:treatment}.

Identification of \(\tau\) therefore amounts to learning enough about
the missing post-treatment no-policy means
\((\mu_T^k(\mathbf d^0))_{k=1}^K\): the treated counterfactual
\(\mu_T^1(\mathbf d^0)\) determines the treatment effect, while the comparison
counterfactuals \((\mu_T^k(\mathbf d^0))_{k=2}^K\) determine the spillover
vector.

The aggregate means may be observed directly or estimated from micro data. In
the latter case, \(\mu_t^k\) denotes the population mean for individuals
attached to unit \(k\) in period \(t\), estimated by the corresponding
unit-period average from a panel or repeated cross section, with survey weights
when appropriate. Identification treats each \(\mu_t^k(\mathbf d)\) as a
population mean and restricts the
missing vector \((\mu_T^k(\mathbf d^0))_{k=1}^K\). Inference addresses sampling
uncertainty when the finite collection of aggregate inputs is estimated.

\noindent\textbf{Notation.} Throughout the paper,
\(\Delta^{K-2}:=\{w\in\mathbb R_+^{K-1}:\sum_{k=2}^K w_k=1\}\) denotes the
simplex of convex comparison weights over units \(k=2,\ldots,K\). The donor
pool is fixed before estimation using eligibility rules specific to the application,
and the baseline comparison domain is the full simplex.
For period changes, write
\(\Delta\mu_t^k(\mathbf d^0):=\mu_t^k(\mathbf d^0)-
\mu_{t-1}^k(\mathbf d^0)\) for \(t=2,\ldots,T_0\), and
\(\Delta\mu_T^k(\mathbf d):=\mu_T^k(\mathbf d)-
\mu_{T_0}^k(\mathbf d^0)\) for the post-treatment target contrast,
\(\mathbf d\in\{\mathbf d^0,\mathbf d^1\}\).
For a vector \(z\), \(\|z\|_1\), \(\|z\|_\infty\), and \(\|z\|\) denote the
\(\ell_1\), sup, and Euclidean norms. The vector \([z]_+\) has \(j\)th
coordinate \(\max\{z_j,0\}\). For a nonempty set \(\mathcal C\),
\(\operatorname{dist}\{z,\mathcal C\}:=\inf_{y\in\mathcal C}\|z-y\|\);
when \(\mathcal C\) is a polyhedron, \(\operatorname{ext}(\mathcal C)\)
denotes its set of extreme points. The symbol \(\mathbbm 1_m\) denotes the
\(m\)-vector of ones, with the subscript omitted when its dimension is clear.
For a matrix, \(\|\cdot\|\) denotes the Euclidean norm after vectorization.
Vector inequalities are interpreted componentwise.

\section{Identification}\label{sec:identification}
\noindent
The identifying mechanism is easiest to see with one treated unit and two
comparison units. Write \(s_2\) and \(s_3\) for the comparison unit spillovers.
As a benchmark with exact validity, suppose the untreated post-treatment change
of the treated unit equals the same weighted change among the comparisons. For
a weight \((\vartheta,1-\vartheta)\) with \(\vartheta\in[0,1]\), the observed
treated minus comparison change then satisfies
\begin{equation}\label{eq:exact_valid_comparison}
\Delta\mu_T^1(\mathbf d^1)
-\vartheta\Delta\mu_T^2(\mathbf d^1)
-(1-\vartheta)\Delta\mu_T^3(\mathbf d^1)
=\tau-\vartheta s_2-(1-\vartheta)s_3 .
\end{equation}
One equation fixes only one linear combination of
\((\tau,s_2,s_3)\). It cannot separate the treatment effect from the two
spillovers. A second valid weight
\((\vartheta',1-\vartheta')\), with \(\vartheta'\in[0,1]\) and
\(\vartheta'\ne\vartheta\), places different weights on the same spillovers
while leaving \((\tau,s_2,s_3)\) unchanged:
\begin{equation}\label{eq:exact_valid_system}
\begin{aligned}
\Delta\mu_T^1(\mathbf d^1)
-\vartheta\Delta\mu_T^2(\mathbf d^1)
-(1-\vartheta)\Delta\mu_T^3(\mathbf d^1)
&=\tau-\vartheta s_2-(1-\vartheta)s_3,\\
\Delta\mu_T^1(\mathbf d^1)
-\vartheta'\Delta\mu_T^2(\mathbf d^1)
-(1-\vartheta')\Delta\mu_T^3(\mathbf d^1)
&=\tau-\vartheta' s_2-(1-\vartheta')s_3 .
\end{aligned}
\end{equation}
The second equation is not new data; its identifying content comes from
maintaining exact validity for both weights. The equations are informative
jointly because they contain the same treatment effect and spillover vector.
Each additional linearly independent comparison further restricts differences
among spillovers. These comparisons cannot determine a common shift. For any
\(\delta\in\mathbb R\), replacing \((\tau,s)\) by
\((\tau+\delta,s+\delta\mathbbm 1_{K-1})\) leaves every expression
\(\tau-w^\top s\) unchanged because every convex weight sums to one. Appendix
Lemma~\ref{lem:exact_comparison_equivalence} formalizes this benchmark.

Exact validity clarifies the role of joint comparison restrictions, but it is
not the maintained empirical model. Weights that achieve exact validity may be
unavailable. For a generic \(w\in\Delta^{K-2}\), the population accounting
identity is
\begin{equation}\label{eq:contaminated_accounting}
\begin{aligned}
\Delta\mu_T^1(\mathbf d^1)-\sum_{k=2}^K w_k\Delta\mu_T^k(\mathbf d^1)
&=
\underbrace{
\Delta\mu_T^1(\mathbf d^0)-\sum_{k=2}^K w_k\Delta\mu_T^k(\mathbf d^0)
}_{\text{latent no-policy gap for }w}
\ +\ \tau-w^\top s .
\end{aligned}
\end{equation}
The first term on the right-hand side is the latent no-policy gap: it is
unobserved because it uses post-treatment means under \(\mathbf d^0\).
The framework first uses each weight's pre-treatment fit to bound this gap. It
then uses restrictions tailored to the application to determine which overall
levels of the treatment effect and spillovers remain possible.

\subsection{Comparison Validity and Relative Effects}
\label{subsec:comparison_validity_identification}

The first restriction converts the accounting identities into a system of
moment inequalities. Each convex weight in the donor simplex receives an
allowance determined by its own pre-treatment fit.

\begin{assumption}[Full-simplex fit-scaled comparison validity]
\label{ass:comparison_validity}
For a fixed envelope constant \(0\le L<\infty\), the untreated potential means
satisfy, for every \(w\in\Delta^{K-2}\),
\[
\left|
\Delta\mu_T^1(\mathbf d^0)
-\sum_{k=2}^K w_k\Delta\mu_T^k(\mathbf d^0)
\right|
\le
\frac{L}{T_0-1}
\sum_{t=2}^{T_0}
\left|
\Delta\mu_t^1(\mathbf d^0)
-\sum_{k=2}^K w_k\Delta\mu_t^k(\mathbf d^0)
\right|
.
\]
\end{assumption}

Once donor eligibility is fixed, the baseline applies the rule to every
nonnegative weight vector that sums to one. Every weight must satisfy its own
restriction for the same treatment effect and spillover vector, while a poorly
fitting weight receives a wider allowance than a well-fitting one. Like
synthetic control and synthetic difference-in-differences
\citep{abadie2010synthetic,arkhangelsky2021synthetic}, the assumption uses
pre-treatment behavior to constrain a post-treatment counterfactual through
maintained restrictions on untreated outcomes. Its distinctive commitment is
uniformity over the prespecified donor simplex. The baseline envelope is also
homogeneous and hence has a zero floor at exact population fit.
Section~\ref{subsec:comparison_validity_interpretation} shows that one
restriction on the factor span can imply validity for all weights and discusses
the implication of the zero floor.

For \(w\in\Delta^{K-2}\), write the observed post-treatment contrast and the
pre-treatment comparison gaps as
\[
\begin{aligned}
C_T^{\mathrm{obs}}(w)
&:=
\Delta\mu_T^1(\mathbf d^1)
-\sum_{k=2}^K w_k\Delta\mu_T^k(\mathbf d^1),\\
C_t(w)
&:=
\Delta\mu_t^1(\mathbf d^0)
-\sum_{k=2}^K w_k\Delta\mu_t^k(\mathbf d^0),
\qquad t=2,\ldots,T_0.
\end{aligned}
\]
Let \(\Pi_0\) denote the true finite-dimensional population input. It collects
the population means entering these contrasts, together with any other input
to the rule defining the admissible set that is estimated from the sampling data.
Prespecified design constants, such as
externally supplied exposure weights or weights based on reference size, are
held fixed and
absorbed into the rule itself rather than included in \(\Pi_0\).
For a treatment effect and spillover vector, set
\begin{equation}\label{eq:relative_effect_vector}
x=(x_2,\ldots,x_K)^\top:=\tau\mathbbm 1_{K-1}-s,
\qquad x_k=\tau-s_k .
\end{equation}
Assumption~\ref{ass:comparison_validity} restricts this vector of relative
effects to
\begin{equation}\label{eq:relative_effect_set}
\mathcal X_L(\Pi_0)
:=
\left\{
x\in\mathbb R^{K-1}:
\left|C_T^{\mathrm{obs}}(w)-w^\top x\right|
\le\frac{L}{T_0-1}\sum_{t=2}^{T_0}|C_t(w)|
\quad\text{for every }w\in\Delta^{K-2}
\right\}.
\end{equation}

The opening benchmark showed why multiple weight directions can be jointly
informative. Validity over the full simplex can also be more informative than
validity imposed only at the donor vertices. When donor pre-treatment gaps offset one
another, an interior weight can have a smaller fit allowance than the
corresponding convex combination of vertex allowances and therefore add a
restriction on \(x\). Under sign alignment this gain disappears. Appendix
Proposition~\ref{prop:comparison_multiplicity} gives the formal result, and
Section~\ref{sec:monte_carlo} illustrates its quantitative importance.

The defining inequalities in \eqref{eq:relative_effect_set} have a direct
geometric interpretation. Each weight defines a slab around the hyperplane for
exact validity, \(C_T^{\mathrm{obs}}(w)=w^\top x\), with tolerance determined by that
weight's pre-treatment fit. Their intersection is
\(\mathcal X_L(\Pi_0)\) in coordinates for the relative effects and,
equivalently, the region allowed by comparison validity in \((\tau,s)\)-space,
as illustrated
in Figure~\ref{fig:comparison_validity_geometry}.

\begin{figure}[!htbp]
\centering
\includegraphics[width=\textwidth]{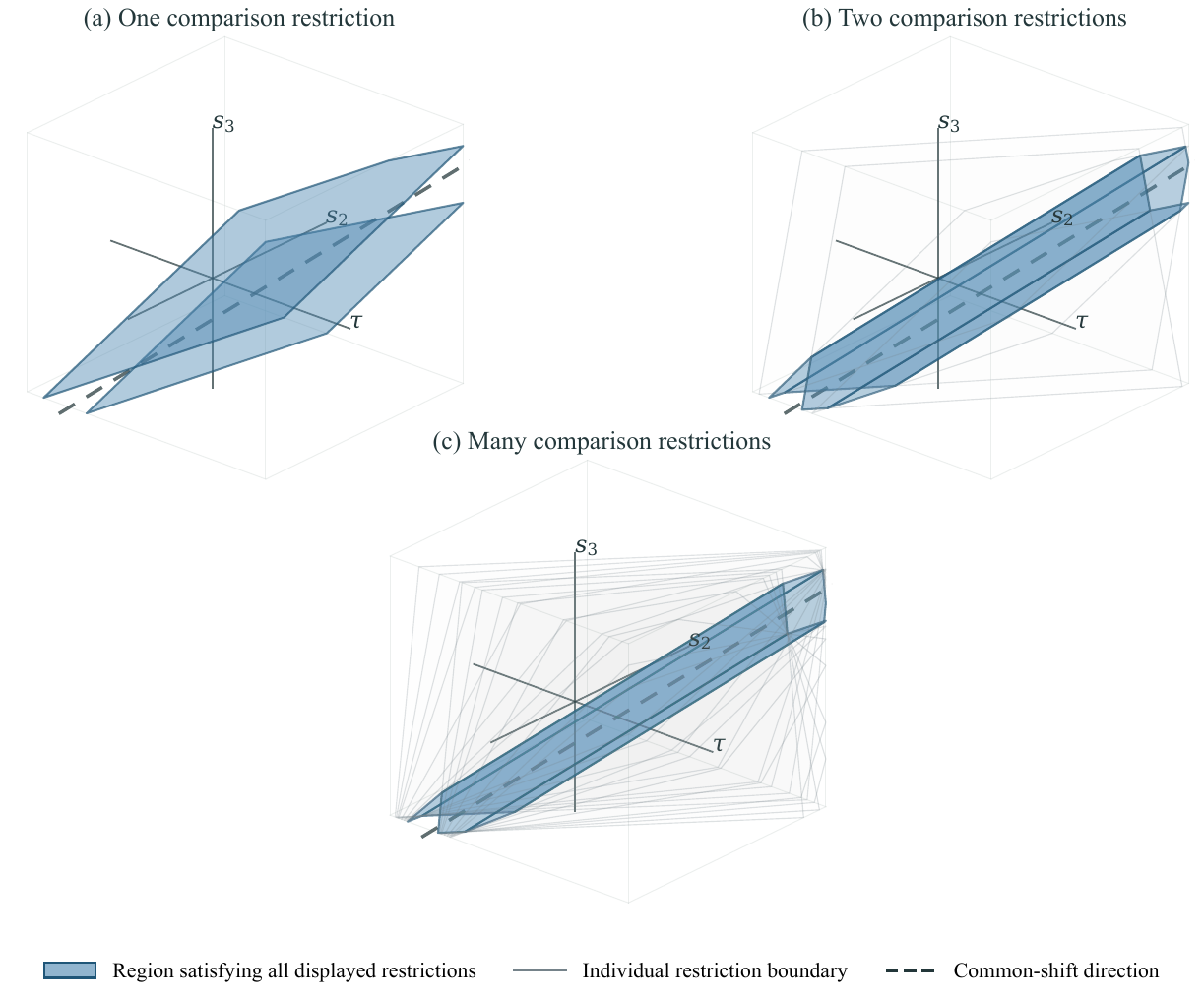}
\caption{Comparison validity as inequalities scaled by pre-treatment fit. Each
comparison weight replaces exact equality with a region whose width depends on
that weight's pre-treatment fit. The panels successively impose one, two, and
many comparison restrictions. In every panel, blue denotes the intersection
that satisfies all displayed restrictions, gray marks the boundaries of the
individual restrictions, and the dashed line indicates the common-shift
direction. The open ends show that adding the same constant to the treatment
effect and all spillovers remains unrestricted.}
\label{fig:comparison_validity_geometry}
\end{figure}

\subsection{From Relative Effects to Treatment Effect Levels}
\label{subsec:sharp_projection}

Comparison validity constrains the treatment effect only relative to
spillovers. If \((\tau,s)\) satisfies all comparison restrictions, then, for
any \(\delta\in\mathbb R\), adding \(\delta\) to \(\tau\) and every component
of \(s\) leaves those restrictions unchanged. This common-shift ambiguity means
that the comparison model alone cannot determine the level of the treatment
effect.

Restricting the common shift does not require an exposure mapping or
unit-specific knowledge of the spillover pattern. To obtain level information,
it is enough to impose prespecified substantive restrictions that limit
movement in this direction. These restrictions may remain agnostic about which
comparison units are affected and how spillovers vary across them. Let \(\Pi\)
denote a generic value of the finite population input, and use
\(\Omega(\Pi)\) to collect these additional restrictions on \((\tau,s)\).

\begin{assumption}[Prespecified admissible rule]
\label{ass:admissible_rule}
The correspondence \(\Pi\mapsto\Omega(\Pi)\subseteq\mathbb R^K\) is fixed
before estimation. Membership in \(\Omega(\Pi)\) exactly summarizes the
additional maintained restrictions on \((\tau,s)\) that are not already part
of comparison validity. At the population input \(\Pi_0\), the set
\(\Omega(\Pi_0)\) is nonempty and closed and contains the population treatment
effect and spillover vector.
\end{assumption}

Formally, let \(\tau^c\in\mathbb R\) be a candidate value of \(\tau\), where
the superscript \(c\) labels the candidate. We call it \emph{compatible},
written \(\mathsf{Comp}(\Pi_0,\tau^c)=1\), when the missing post-treatment
no-policy means and spillovers can be completed under comparison validity and
the prespecified admissible rule with \(\tau=\tau^c\).

\begin{proposition}[Treatment effect projection]
\label{prop:treatment_effect_projection}
Suppose Assumptions~\ref{ass:comparison_validity} and
\ref{ass:admissible_rule} hold. Then, for every \(\tau^c\in\mathbb R\),
\begin{equation}
\mathsf{Comp}(\Pi_0,\tau^c)=1
\quad\Longleftrightarrow\quad
\exists x\in\mathcal X_L(\Pi_0):
\ (\tau^c,\tau^c\mathbbm 1_{K-1}-x)\in\Omega(\Pi_0).
\label{eq:sharp_completion_general}
\end{equation}
Consequently, the identified set for the treatment effect is
\begin{equation}\label{eq:sharp_anchor_union}
\begin{aligned}
\Theta_\tau(\Pi_0;L)
&=\{\tau^c\in\mathbb R:\mathsf{Comp}(\Pi_0,\tau^c)=1\}\\
&=
\bigcup_{x\in\mathcal X_L(\Pi_0)}
\left\{
\tau^c\in\mathbb R:
(\tau^c,\tau^c\mathbbm 1_{K-1}-x)\in\Omega(\Pi_0)
\right\}.
\end{aligned}
\end{equation}
Every value in the displayed set is attained by a completion of the
maintained model of potential means.
\end{proposition}

\begin{proof}
See \hyperref[proof:treatment_effect_projection]{Appendix~\ref*{app:candidate_system_representation_proofs}}.
\end{proof}

The result separates the roles of the two parts of the model.
\(\mathcal X_L(\Pi_0)\) contains the vectors of relative effects allowed by the
comparison restrictions, while \(\Omega(\Pi_0)\) determines which treatment
effect levels can accompany each vector. If \(\Omega(\Pi_0)\) leaves the
common-shift direction unrestricted, the treatment effect remains unbounded in
that direction. Ruling out movement in one or both directions produces one-
or two-sided bounds.

Because every value in the projection is attained,
\(\Theta_\tau(\Pi_0;L)\) is sharp relative to the maintained model in
Assumptions~\ref{ass:comparison_validity} and \ref{ass:admissible_rule}. In
particular, \(\Omega(\Pi_0)\) is the complete collection of additional
restrictions on \((\tau,s)\) beyond comparison validity. When the dependence
on a data-generating process (DGP) \(P\) is explicit, write \(\Pi(P)\) for
its population input and
\(\Theta_\tau(P;L):=\Theta_\tau(\Pi(P);L)\).

\subsection{Finite Representation and Treatment Effect Bounds}
\label{subsec:candidate_solvability}

Proposition~\ref{prop:treatment_effect_projection} characterizes candidate
compatibility in terms of \(\mathcal X_L(\Pi_0)\) and \(\Omega(\Pi_0)\), but
that characterization is not yet a finite computational problem:
\(\mathcal X_L(\Pi_0)\) contains a continuum of restrictions, one for each
weight,
while \(\Omega(\Pi_0)\) has so far been left abstract. The next lemma resolves
the first issue. It replaces the two-sided restriction for every \(w\) in the
donor simplex exactly by \(2(K-1)\) linear donor inequalities and box
constraints on two \((T_0-1)\)-dimensional auxiliary vectors.
For \(k=2,\ldots,K\), let \(e_k\in\Delta^{K-2}\) denote the simplex vertex
placing all weight on donor \(k\); hence \(C_t(e_k)\) and
\(C_T^{\mathrm{obs}}(e_k)\) are the corresponding contrasts for individual donors.

\begin{lemma}[Exact finite representation of validity over the full simplex]
\label{lem:finite_comparison_bridge}
For \(L\ge0\), \(x\in\mathcal X_L(\Pi_0)\) if and only if there exist
\[
v^\pm=(v_2^\pm,\ldots,v_{T_0}^\pm)^\top
\in
\left[-\frac{L}{T_0-1},\frac{L}{T_0-1}\right]^{T_0-1}
\]
such that
\begin{equation}\label{eq:finite_comparison_bridge}
\sum_{t=2}^{T_0}v_t^-C_t(e_k)
\le C_T^{\mathrm{obs}}(e_k)-x_k
\le \sum_{t=2}^{T_0}v_t^+C_t(e_k),
\qquad k=2,\ldots,K.
\end{equation}
\end{lemma}

\begin{proof}
See \hyperref[proof:finite_comparison_bridge]{Appendix~\ref*{app:candidate_system_representation_proofs}}.
\end{proof}

The reduction follows from the identity for the support function
\[
\frac{L}{T_0-1}\sum_{t=2}^{T_0}|C_t(w)|
=
\max_{\substack{v\in\mathbb R^{T_0-1}\\
\lVert v\rVert_\infty\le L/(T_0-1)}}
\sum_{t=2}^{T_0}v_tC_t(w).
\]
This identity replaces the fit allowance involving absolute values by a linear
expression in bounded auxiliary coordinates. Because the donor simplex and
the coefficient box are compact and convex and the criterion is bilinear,
Sion's minimax theorem permits the order of optimization to be interchanged.
This interchange yields one auxiliary vector that certifies the restriction
simultaneously for all donor weights. Conditional on that vector, both the
residual after treatment and the lifted allowance are linear in \(w\), so their
maximum over the simplex is attained at a donor vertex \(e_k\). The two signs
of the original restriction involving absolute values require the separate
certificates \(v^+\) and \(v^-\). These vectors are not economic parameters: together with
the finite donor rows in \eqref{eq:finite_comparison_bridge}, they certify all
restrictions for interior weights simultaneously.

Lemma~\ref{lem:finite_comparison_bridge} makes the comparison validity
component finite, but the admissible set
\(\Omega(\Pi)\) remains abstract under
Assumption~\ref{ass:admissible_rule}, which requires neither polyhedrality nor
a particular representation. To turn the entire compatibility problem into a
finite linear system, impose the following additional structure.

\begin{assumption}[Finite linear representation of admissible restrictions]
\label{ass:finite_admissible_representation}
The admissible rule in Assumption~\ref{ass:admissible_rule} has the
representation
\[
\Omega(\Pi)=
\left\{(\tau,s):\exists u\ \text{such that}\quad
A_\tau(\Pi)\tau+A_s(\Pi)s+A_u u\le a(\Pi)
\right\},
\]
where \(A_\tau(\Pi)\), \(A_s(\Pi)\), and \(a(\Pi)\) are conformable affine
functions of \(\Pi\), \(A_u\) is fixed, and \(g\) is a fixed vector satisfying
\begin{equation}
A_\tau(\Pi)+A_s(\Pi)\mathbbm 1_{K-1}=g
\quad\text{for every admissible }\Pi.
\label{eq:fixed_common_shift_loading}
\end{equation}
The auxiliary vector \(u\) has fixed dimension.
An exact equality is represented by two prespecified opposite inequality
rows. All dimensions, coordinates, and units are fixed before estimation.
\end{assumption}

When no auxiliary lift is needed, omit \(u\) and use the direct half-space
representation. The vector \(u\) otherwise accommodates linear lifts, such as
the auxiliary coordinates used to represent restrictions involving absolute
values.

Assumption~\ref{ass:finite_admissible_representation} does two things needed
for computation. Its finite inequality representation turns membership in
\(\Omega(\Pi)\) into finitely many rows. Its common-shift condition ensures
that the candidate \(\tau^c\) enters those rows through the fixed loading
\(g\). Indeed, substituting \(s=\tau^c\mathbbm 1_{K-1}-x\) gives
\begin{equation}
-A_s(\Pi)x+A_u u
\le a(\Pi)-g\tau^c,
\label{eq:common_shift_candidate_rows}
\end{equation}
so the coefficients on the relative effects may depend on estimated population
inputs, but the coefficient on \(\tau^c\) does not. Consequently, estimation
error enters the system in the same way for every candidate value.\footnote{This
structure permits, for example, estimated normalized
exposure weights whose row sums are fixed.} The projection result itself
depends only on the set \(\Omega(\Pi_0)\); the conditions here provide a finite
representation suited to computation and inference.

At a fixed candidate value \(\tau^c\), the preceding two steps now yield one
finite system. Lemma~\ref{lem:finite_comparison_bridge} supplies the comparison
and box rows, while
\eqref{eq:common_shift_candidate_rows} supplies the admissibility rows. The
coefficients \(C_t(e_k)\) and \(C_T^{\mathrm{obs}}(e_k)\) are affine functions
of the population input \(\Pi\) and are evaluated at that input below. Stack
the unknown coordinates as
\[
\eta=(x^\top,u^\top,(v^-)^\top,(v^+)^\top)^\top,
\]
omitting \(u\) when it is unnecessary. The candidate rows are
\begin{equation}\label{eq:bridge_candidate_master_system}
\begin{aligned}
\sum_{t=2}^{T_0}v_t^-C_t(e_k)+x_k
&\le C_T^{\mathrm{obs}}(e_k),
\qquad k=2,\ldots,K,\\
-\sum_{t=2}^{T_0}v_t^+C_t(e_k)-x_k
&\le-C_T^{\mathrm{obs}}(e_k),
\qquad k=2,\ldots,K,\\
-\frac{L}{T_0-1}\mathbbm 1_{T_0-1}\le v^-
&\le\frac{L}{T_0-1}\mathbbm 1_{T_0-1},
\qquad
-\frac{L}{T_0-1}\mathbbm 1_{T_0-1}\le v^+
\le\frac{L}{T_0-1}\mathbbm 1_{T_0-1},\\
-A_s(\Pi)x+A_u u
&\le a(\Pi)-g\tau^c.
\end{aligned}
\end{equation}
Let \(A(\Pi)\) be the coefficient matrix and \(b(\Pi,\tau^c)\) the offset
obtained by stacking the displayed inequalities in the form
\(b(\Pi,\tau^c)-A(\Pi)\eta\le0\). With the displayed coordinates and row order,
write the solution set as
\begin{equation}
\mathcal F(\Pi,\tau^c)
:=
\{\eta:b(\Pi,\tau^c)-A(\Pi)\eta\le0\}.
\label{eq:identification_candidate_feasible_set}
\end{equation}

The following theorem completes the reduction: a candidate is compatible with
the original model of potential means exactly when this linear system has a
solution.

\begin{theorem}[Exact finite candidate representation]
\label{prop:candidate_tau_linear_system}
Suppose Assumptions~\ref{ass:comparison_validity} and
\ref{ass:admissible_rule}--\ref{ass:finite_admissible_representation} hold.
Then, for every \(\tau^c\in\mathbb R\),
\begin{equation}
\mathsf{Comp}(\Pi_0,\tau^c)=1
\quad\Longleftrightarrow\quad
\mathcal F(\Pi_0,\tau^c)\ne\varnothing.
\label{eq:sharp_completion_finite_solvability}
\end{equation}
The finite system introduces neither discretization nor relaxation of the
comparison domain. Every feasible solution can be lifted to a completion of
the maintained model of potential means.
\end{theorem}

\begin{proof}
See \hyperref[proof:candidate_tau_linear_system]{Appendix~\ref*{app:candidate_system_representation_proofs}}.
\end{proof}

Under Assumption~\ref{ass:finite_admissible_representation}, the identified
set in Proposition~\ref{prop:treatment_effect_projection} is also the scalar
projection
\[
\Theta_\tau(\Pi_0;L)
=
\{\tau^c\in\mathbb R:\mathcal F(\Pi_0,\tau^c)\ne\varnothing\}
\]
of the finite candidate system.

\begin{corollary}[Shape of the identified set]
\label{cor:identified_set_shape}
Under the assumptions of
Theorem~\ref{prop:candidate_tau_linear_system}, the population identified set
\(\Theta_\tau(\Pi_0;L)\) is a nonempty polyhedron in \(\mathbb R\). Hence it is
a closed interval, possibly a singleton, a closed half-line, or all of
\(\mathbb R\).
\end{corollary}

\begin{proof}
The joint set
\[
\{(\tau^c,\eta):b(\Pi_0,\tau^c)-A(\Pi_0)\eta\le0\}
\]
is a polyhedron. It is nonempty because the population treatment effect is
compatible under the maintained assumptions. Its projection onto the scalar
\(\tau^c\)-coordinate is therefore a polyhedron in \(\mathbb R\), and every
nonempty polyhedron in \(\mathbb R\) has one of the stated forms.
\end{proof}

The corollary characterizes the shape of the population identified set but
does not determine whether it is bounded above or below. Comparison validity
is invariant to a common shift in the treatment effect and all spillovers, so
boundedness depends on whether the admissible set blocks movement in that
direction. The finite representation in
Assumption~\ref{ass:finite_admissible_representation} does not itself block
this movement; it makes the resulting compatibility problem computationally
tractable. The following proposition gives the exact criterion.

\begin{proposition}[Treatment effect anchoring by admissible restrictions]
\label{prop:exposure_anchoring}
Suppose Assumptions~\ref{ass:comparison_validity} and
\ref{ass:admissible_rule}--\ref{ass:finite_admissible_representation} hold.
Write \(\operatorname{rec}\{\Omega(\Pi_0)\}\) for the recession cone of
\(\Omega(\Pi_0)\), and let
\(\iota=(1,\mathbbm 1_{K-1}^\top)^\top\) denote the direction of a common
shift. The treatment effect
projection is bounded above if and only if
\(\iota\notin\operatorname{rec}\{\Omega(\Pi_0)\}\), bounded below if and only if
\(-\iota\notin\operatorname{rec}\{\Omega(\Pi_0)\}\), and bounded on both sides if
and only if
\[
\operatorname{rec}\{\Omega(\Pi_0)\}\cap\operatorname{span}\{\iota\}=\{0\}.
\]
\end{proposition}

\begin{proof}
See \hyperref[proof:exposure_anchoring]{Appendix~\ref*{app:candidate_system_representation_proofs}}.
\end{proof}

For a direct admissible row,
the corresponding component of
\(g=A_\tau(\Pi)+A_s(\Pi)\mathbbm 1_{K-1}\) is its loading on this direction.
A zero loading means that the row restricts only relative effects; a one-sided
row with nonzero loading may block one ray, and two opposing rows may block
both.\footnote{With auxiliary variables, a recession direction in the lift can
offset an individual row's direct loading, so the recession cone of the
complete admissible set remains the relevant criterion.}

These population results do not impose the same shape on the finite-sample
test-inversion set because the candidate-specific bootstrap critical value may
vary with \(\tau^c\). Compatibility at any fixed candidate is a feasibility
problem for a linear program (LP). Across candidates, the coefficient matrix
remains fixed and \(\tau^c\) changes only the affine right-hand side, which
permits computational reuse and ensures that estimation error enters in the
same way across candidates. The next subsection turns to the interpretation
and calibration of \(L\). Section~\ref{sec:inference} then introduces the
corresponding Farkas certificate and uses it to test fixed candidates.

\subsection{Interpretation and Calibration of \texorpdfstring{\(L\)}{L}}
\label{subsec:comparison_validity_interpretation}

The preceding results characterize identification at a fixed envelope value.
Once \(L\) is fixed, Assumption~\ref{ass:comparison_validity} converts each
weight's pre-treatment discrepancy into an allowance for its missing
post-treatment no-policy gap. The smallest allowance required by the missing
no-policy outcomes is not identified from observed data. I therefore treat
\(L\) as a prespecified sensitivity parameter. Each fixed value defines a
separate maintained model, and the sensitivity path records how the compatible
treatment effects change as more post-treatment discrepancy is allowed.

This creates an interpretation problem: a sensitivity path is useful only if
values such as \(L=1\) or \(L=2\) have a meaningful scale. I use
leave-one-period-out pre-treatment exercises for this purpose. Each observed
pre-treatment period is treated in turn as a pseudo-post period, and the
smallest envelope needed to accommodate that movement is computed. These
placebo values provide empirical benchmarks for the sensitivity scale; they do
not select the post-treatment envelope.

To formalize this distinction, recall the pre-treatment comparison gaps
\(C_t(w)\), and define the latent
post-treatment gap
\[
C_T^0(w)
:=
\Delta\mu_T^1(\mathbf d^0)
-\sum_{k=2}^K w_k\Delta\mu_T^k(\mathbf d^0).
\]
Assumption~\ref{ass:comparison_validity} links fit before treatment to the
latent gap after treatment:
\[
|C_T^0(w)|
\le
\frac{L}{T_0-1}\sum_{t=2}^{T_0}|C_t(w)|
\qquad\text{for every }w\in\Delta^{K-2}.
\]

\begin{definition}[Latent threshold and calibration using observed periods]
\label{def:observed_period_calibration}
Suppose \(T_0\ge3\).
The smallest envelope value required by the latent post-treatment gaps is
\[
L_T^\star
:=
\inf\left\{
\bar L\ge0:
|C_T^0(w)|
\le
\frac{\bar L}{T_0-1}
\sum_{t=2}^{T_0}|C_t(w)|
\quad\text{for every }w\in\Delta^{K-2}
\right\}.
\]
For a pre-treatment period \(\ell\in\{2,\ldots,T_0\}\), define its placebo
calibration index by leaving that period out:
\[
L_\ell^{\mathrm{pl}}
:=
\inf\left\{
\bar L\ge0:
|C_\ell(w)|
\le
\frac{\bar L}{T_0-2}
\sum_{\substack{t=2\\t\ne\ell}}^{T_0}|C_t(w)|
\quad\text{for every }w\in\Delta^{K-2}
\right\}.
\]
The infimum is \(+\infty\) if the corresponding set is empty.
\end{definition}

Assumption~\ref{ass:comparison_validity} at a fixed \(L\) holds if and only if
\(L\ge L_T^\star\). The threshold \(L_T^\star\) depends on missing no-policy
means and is not identified, whereas each \(L_\ell^{\mathrm{pl}}\) depends only
on pre-treatment population means. Following \citet{hsu2013calibrating}, the
placebo indices express candidate values of \(L\) in units of observed
pre-treatment instability. They neither estimate nor automatically bound
\(L_T^\star\), and choosing the sensitivity range remains a substantive
decision.

A low-rank factor model provides one primitive interpretation of the common
envelope.
Suppose untreated first differences satisfy
\[
\Delta\mu_t^k(\mathbf d^0)=\alpha_t+f_t^\top h_k,
\qquad t=2,\ldots,T_0,T,\quad k=1,\ldots,K,
\]
where \(\alpha_t\) is a common increment, \(f_t\) is a factor vector, and
\(h_k\) is unit \(k\)'s loading vector. Write
\(F_{\mathrm{pre}}=[f_2,\ldots,f_{T_0}]\). If
\(f_T=F_{\mathrm{pre}}\gamma\) for some
\(\gamma\in\mathbb R^{T_0-1}\), then the common increment \(\alpha_t\) cancels
and, simultaneously for every \(w\),
\[
|C_T^0(w)|
\le
\|\gamma\|_\infty
\sum_{t=2}^{T_0}|C_t(w)|.
\]
Thus Assumption~\ref{ass:comparison_validity} holds for all weights whenever
\(L\ge(T_0-1)\|\gamma\|_\infty\). A single restriction on factor movements
can imply validity for every weight. Applying the same argument with an
observed pre-treatment period held out gives the placebo counterpart. The
factor structure can rationalize a common envelope across weights, but it does
not identify the envelope required after treatment. Appendix
Example~\ref{ex:factor_comparison_validity} gives both derivations.

The envelope is positively homogeneous, so a population weight with exact fit
receives zero allowance for its latent post-treatment gap. The
relative-magnitude restriction of \citet{rambachan2023more}, although it bounds
changes in parallel trends violations rather than latent gaps for individual
weights, has the same implication. Under the condition on the factor span,
zero factor gaps before treatment force the gap after treatment to zero.
Appendix Remark~\ref{rem:additive_floor_extension} develops a relaxation
with an additive floor.

The baseline score averages absolute annual pre-treatment discrepancies, so
\(L\) bounds a weight's latent post-treatment gap relative to its typical
annual discrepancy. When the post-treatment target spans a different horizon,
as in the LAWA application, \(L\) also accounts for extrapolation between the
different horizons. The placebo indices based on smoothed factors in the application provide
annual benchmarks on this scale.\footnote{An application may
substitute another prespecified fit score with a clear interpretation.
Appendix Remark~\ref{rem:additive_floor_extension} gives an exact finite
representation for a nonnegative fit-independent floor. Holding the remaining
specification fixed, increasing this floor weakly enlarges the population
compatibility set.}

The finite representation depends on the prespecified collection of pre-treatment gap
coordinates \(C_2(w),\ldots,C_{T_0}(w)\), not on annual first differences per
se.\footnote{The coordinates may instead be constructed from levels, long
differences, residualized gaps, or horizon-matched contrasts.} Changing them
changes the economic content and scale of \(L\), but not the finite representation or the
geometry of relative effects. The baseline based on annual first differences provides eight
transparent pre-treatment movements for both the validity score and its
placebo calibration.

This scalar envelope belongs to the broader literature on partial identification,
which weakens point-identifying benchmark restrictions in explicit and
interpretable ways \citep[e.g.,][]{manski2000monotone,manski2018right,
rambachan2023more}. Here the restriction bounds a latent gap for each weight in
equations that also contain the unknown spillover vector \(s\).

\section{Inference}\label{sec:inference}

Both identification and inference ask whether a given value \(\tau^c\) is
compatible with the model. The null does not compare \(\tau^c\) with an
estimated endpoint. It asks whether some values of the unobserved no-policy
means and spillovers satisfy all assumptions when \(\tau=\tau^c\).
Theorem~\ref{prop:candidate_tau_linear_system} turns this into checking whether
a finite linear system has a solution, and Farkas' alternative provides a
certificate when it does not. I use the normalized solvability statistic and
bootstrap calibration of \citet{goff2025inference} for this structured family,
establish uniform validity under model-tailored regularity conditions, and
invert the tests over candidate values. The systems share a candidate-invariant
primitive sampling perturbation. A uniform feasibility error bound is imposed
only over model-generated compatible systems, and the variance condition is
restricted to structurally reachable certificate directions that can determine
rejection.

The donor pool, envelope value \(L\), admissible rule, and finite row
representation are held fixed throughout. Section~\ref{sec:candidate_lp_anatomy}
constructs a test for one candidate by turning sample infeasibility into a
normalized Farkas score and calibrating it with the bootstrap.
Section~\ref{sec:test_inversion_over_candidates} then inverts these tests, and
Section~\ref{subsec:uniform_validity} states their candidatewise coverage.

\subsection{Compatibility Test for a Fixed Candidate}
\label{sec:candidate_lp_anatomy}

For a fixed candidate \(\tau^c\), inference asks whether the sample evidence of
incompatibility is larger than estimation error can explain. Recall that
\(\eta\) collects the relative effects and auxiliary variables in the finite
representation. By Theorem~\ref{prop:candidate_tau_linear_system}, the
candidate is compatible at a DGP \(P\) exactly when
\(b(\Pi(P),\tau^c)-A(\Pi(P))\eta\le0\) has a solution. For sample size \(n\),
let \(\hat\Pi_n\) be the sample analogue of \(\Pi(P)\), and write
\(\hat A_n:=A(\hat\Pi_n)\) and
\(\hat b_n(\tau^c):=b(\hat\Pi_n,\tau^c)\).

Farkas' alternative turns failure of this system into a certificate. To see
why, suppose the population system were feasible. Multiplying its rows by any
\(\lambda\ge0\) and summing gives
\(b(\Pi(P),\tau^c)^\top\lambda
-\eta^\top A(\Pi(P))^\top\lambda\le0\). If
\(A(\Pi(P))^\top\lambda=0\), feasibility therefore requires
\(b(\Pi(P),\tau^c)^\top\lambda\le0\). Farkas' alternative states that the
converse also holds:
\begin{equation}
\mathcal F(\Pi(P),\tau^c)=\varnothing
\quad\Longleftrightarrow\quad
\exists\lambda\ge0:\
A(\Pi(P))^\top\lambda=0,\quad
b(\Pi(P),\tau^c)^\top\lambda>0.
\label{eq:farkas_candidate_alternative}
\end{equation}
Thus \(\lambda\) is a nonnegative combination of the system rows that cancels
the unknown \(\eta\) and leaves a contradiction.

The sample system replaces \(\Pi(P)\) by \(\hat\Pi_n\). Even when the
population system is feasible, it may lie on the boundary, where arbitrarily
small estimation error can make the sample system infeasible. Sample
infeasibility alone is therefore not valid evidence against compatibility.
Moreover, a Farkas certificate is scale-free: multiplying \(\lambda\) by a
positive constant gives the same contradiction. An unweighted normalization
would make the score depend on the arbitrary units used to write individual
inequality rows. Instead, let \(\widehat{\mathsf S}_n\) scale each row by the
largest bootstrap-estimated standard deviation among its estimated entries;
fixed rows receive zero scale. Weighting by these row scales makes the score
invariant to positive rescaling of a row and expresses a contradiction relative
to its sampling uncertainty. The normalized sample score is
\begin{equation}\label{eq:sample_uniform_infeasibility}
\hat Q_n(\tau^c)
:=
\sup_{\substack{\lambda\ge0:\ 
\hat A_n^\top\lambda=0\\
\mathbbm 1^\top\widehat{\mathsf S}_n\lambda\le1}}
\hat b_n(\tau^c)^\top\lambda.
\end{equation}
It is zero when the sample system is feasible, positive when a normalized
certificate separates it, and may be infinite when fixed rows alone certify
incompatibility.

The bootstrap cannot simply treat the estimated system as the null system. A
compatible population system may lie on the boundary of feasibility, so its
sample analogue can be infeasible with nonvanishing probability. The procedure
therefore recenters the bootstrap perturbation of the estimated system at a
near-feasible completion and evaluates it over certificate directions that are
nearly optimal for the candidate. Its bootstrap distribution measures how
large the normalized infeasibility score can be when the candidate remains
compatible.

Fix a nominal test level \(\alpha\in(0,1/2)\). The bootstrap estimates how
large the score can be from sampling error at a compatible boundary. The
following procedure summarizes the test; Appendix equations
\eqref{eq:uniform_eta_hat}--\eqref{eq:uniform_bootstrap_stat} give the exact
recentering and certificate calculations.

\begin{procedure}[Compatibility test for a fixed candidate]
\label{proc:fixed_candidate_compatibility}
Within the fixed specification, choose \(\tau^c\in\mathbb R\).
\begin{enumerate}[label=\arabic*.]
\item Construct the sample system
\(\hat b_n(\tau^c)-\hat A_n\eta\le0\) and compute
\(\hat Q_n(\tau^c)\).
\item Construct the near-feasible completion, perturb the estimated system
with the bootstrap, and maximize over the normalized certificates that are
nearly optimal for this candidate. Let \(\hat c_n(\tau^c,1-\alpha)\) be the
resulting conditional critical value.
\item Do not reject compatibility when
\[
\hat Q_n(\tau^c)<\infty
\quad\text{and}\quad
\sqrt n\,\hat Q_n(\tau^c)\le\hat c_n(\tau^c,1-\alpha).
\]
\end{enumerate}
\end{procedure}

\subsection{Test Inversion over Candidate Values}
\label{sec:test_inversion_over_candidates}

Let \(\mathcal T\subset\mathbb R\) be a compact reporting domain. It limits
the range over which the test is inverted; it does not enter the maintained
model or the population identified set. Inverting
Procedure~\ref{proc:fixed_candidate_compatibility} over this domain gives
the inversion set
\begin{equation}\label{eq:uniform_inverted_cs}
\mathcal I_n
:=
\left\{
\tau^c\in\mathcal T:
\hat Q_n(\tau^c)<\infty,\quad
\sqrt n\,\hat Q_n(\tau^c)
\le \hat c_n(\tau^c,1-\alpha)
\right\}.
\end{equation}

This is direct inversion of compatibility tests rather than inference based on
plug-in estimates of projection endpoints. Its coverage interpretation is
candidatewise: each compatible structural value is protected against false
exclusion by its own level-\(\alpha\) test. Section~\ref{subsec:uniform_validity}
states the corresponding uniform result.

The model structure also allows
substantial computational reuse within a fixed specification. The matrix from
the finite representation, row scales, and underlying bootstrap perturbations are constructed
once. Since \(\tau^c\) changes only the affine right-hand side, the feasible region
of the sample Farkas LP is common across candidates and only its objective
changes. At a given candidate, the bootstrap certificate LP has a common
feasible region across draws and only its objective changes.
The critical value nevertheless remains candidate-specific because the
near-feasible completion and near-optimal certificate set depend on
\(\tau^c\).

The set in \eqref{eq:uniform_inverted_cs} is defined by continuum inversion. In
practice, the application approximates it on an adaptive grid.\footnote{The
coverage result applies to each compatible candidate that is actually tested;
interpolation between tested values is a numerical display, not an additional
coverage claim. Appendix~\ref{app:arizona_policy_coding} records the grid,
tolerance, solver, search anchors, and audit using a denser grid.}

\subsection{Candidatewise Validity}\label{subsec:uniform_validity}

Let \(\mathcal P^U\) denote the class of DGPs over which the tests are required
to be uniformly valid. The attainment result in
Proposition~\ref{prop:treatment_effect_projection} implies that every compatible
candidate value can be the structural effect in a completion of the maintained
population model.
Accordingly, define the null index class
\[
\mathcal Q_0
:=
\{(P,\tau^c):P\in\mathcal P^U,\ \tau^c\in\mathcal T,\
\ \mathcal F(\Pi(P),\tau^c)\ne\varnothing\}.
\]
Write \(\PP_P\) for probability under \(P\). The inferential target is
\begin{equation}
\limsup_{n\to\infty}
\sup_{(P,\tau^c)\in\mathcal Q_0}
\PP_P\{\tau^c\notin\mathcal I_n\}
\le\alpha.
\label{eq:robust_candidate_coverage_target}
\end{equation}
This controls false exclusion whichever compatible value is the structural
truth.

\begin{assumption}[First-order sampling of aggregate inputs]
\label{ass:sampling}
The dimension of \(\Pi(P)\) is fixed and \(\Pi(P)\) lies in a fixed compact
set. Let \(\hat\Pi_n^*\) denote the input recomputed under the bootstrap, and
define
\[
Z_{n,P}:=\sqrt n\{\hat\Pi_n-\Pi(P)\},
\qquad
Z_n^*:=\sqrt n\{\hat\Pi_n^*-\hat\Pi_n\}.
\]
Uniformly over \(P\in\mathcal P^U\), \(Z_{n,P}\) admits a bounded-Lipschitz
Gaussian approximation with a uniformly bounded, possibly singular covariance
matrix \(\Sigma_\Pi(P)\). Conditionally on the data, \(Z_n^*\) consistently
reproduces this law and covariance, uniformly in probability.
\end{assumption}

Appendix Lemma~\ref{lem:cluster_ratio_sampling_sufficient} gives primitive
sufficient conditions for weighted aggregate means formed from independent
sampling clusters. Inputs observed without sampling error enter the candidate
system as fixed components.

The finite row representation is affine in \(\Pi\), while \(\tau^c\) has a
fixed loading on its right-hand side. The population, sample, and bootstrap
systems therefore share the same first-order sampling perturbation at every
candidate; indexing the test by \(\tau^c\) introduces no additional stochastic
equicontinuity requirement.

For each system row \(r\), let \(\sigma_r(P)\) be the largest asymptotic
standard deviation of the \(\sqrt n\)-scaled estimation errors among the
data-estimated entries in row \(r\) of
\([\hat A_n\ \ \hat b_n(\tau^c)]\). Set \(\sigma_r(P)=0\) when the row
contains only fixed entries, and define
\[
\mathsf S(P):=\operatorname{diag}\{\sigma_r(P)\}_r.
\]
These row scales do not depend on \(\tau^c\). Because the system has finitely
many entries with fixed loadings on \(\Pi\), the covariance bound in
Assumption~\ref{ass:sampling} implies a common finite upper bound \(\bar d\)
on their asymptotic standard deviations. Fixed entries have zero scale by
construction. The population analogues of the normalized dual region and
sample score are
\[
\begin{aligned}
\mathcal D(P)
&:=
\{\lambda\ge0:A(\Pi(P))^\top\lambda=0,
\ \mathbbm 1^\top\mathsf S(P)\lambda\le1\},\\
Q(P,\tau^c)
&:=\sup_{\lambda\in\mathcal D(P)}
b(\Pi(P),\tau^c)^\top\lambda.
\end{aligned}
\]
At a compatible candidate, \(Q(P,\tau^c)=0\); a nonzero maximizer is a
certificate direction that is potentially binding at the feasibility
boundary.

To describe sampling variation along these directions, let
\(\eta^0(P,\tau^c)\) be the minimum-norm feasible nuisance vector and write
\(D_\Pi\) for differentiation with respect to \(\Pi\). Define the first-order
sensitivity of the row residuals at this completion by
\[
\Gamma(P,\tau^c)
:=
\left.
D_\Pi\!\left[
b(\Pi,\tau^c)-A(\Pi)\eta^0(P,\tau^c)
\right]
\right|_{\Pi=\Pi(P)}.
\]
\begin{assumption}[Regular candidate systems]
\label{ass:regular_candidate_systems}
There exist constants
\(\underline d,c_\Gamma,\delta_\Gamma,\varepsilon_\Gamma>0\) and
\(\bar H<\infty\), common to the relevant DGP and null classes, with
\(\delta_\Gamma\le\min\{\underline d/2,\bar d\}\), such that:
\begin{enumerate}[label=(\roman*)]
\item every system entry classified as estimated has asymptotic standard
deviation at least \(\underline d\);
\item for every \((P,\tau^c)\in\mathcal Q_0\) and every \(\eta\),
\[
\operatorname{dist}\{\eta,\mathcal F(\Pi(P),\tau^c)\}
\le
\bar H\left\|[b(\Pi(P),\tau^c)-A(\Pi(P))\eta]_+\right\|.
\]
\item for every \((P,\tau^c)\in\mathcal Q_0\), if the population program
defining \(Q(P,\tau^c)\) has a nonzero
maximizer, at least one such maximizer \(\lambda\) satisfies
\[
\lambda^\top\Gamma(P,\tau^c)\Sigma_\Pi(P)
\Gamma(P,\tau^c)^\top\lambda\ge c_\Gamma.
\]
For each \(\widetilde\Pi\) in the input domain such that
\(\widetilde\Pi-\Pi(P)\) is supported on the coordinates estimated from the
sampling data, and each diagonal \(\widetilde{\mathsf S}\) having the same
zero pattern as \(\mathsf S(P)\), define
\[
\widetilde A:=A(\widetilde\Pi),
\qquad
\widetilde b:=b(\widetilde\Pi,\tau^c),
\qquad
\widetilde{\mathcal D}
:=
\{\lambda\ge0:\widetilde A^\top\lambda=0,
\ \mathbbm 1^\top\widetilde{\mathsf S}\lambda\le1\}.
\]
If zero is the unique maximizer of the population program, then, for every
such \((\widetilde\Pi,\widetilde{\mathsf S})\) satisfying
\[
\|\widetilde A-A(\Pi(P))\|
+\|\widetilde b-b(\Pi(P),\tau^c)\|
+\|\widetilde{\mathsf S}-\mathsf S(P)\|
\le\delta_\Gamma,
\]
every
\(\lambda\in\operatorname{ext}(\widetilde{\mathcal D})\setminus\{0\}\)
satisfying \(\widetilde b^\top\lambda\ge-\varepsilon_\Gamma\) obeys
\[
\lambda^\top\Gamma(P,\tau^c)\Sigma_\Pi(P)
\Gamma(P,\tau^c)^\top\lambda\ge c_\Gamma.
\]
\end{enumerate}
\end{assumption}

Part (i) rules out vanishing first-order variation in entries classified as
estimated; the restriction on \(\delta_\Gamma\) preserves this nondegeneracy
for nearby normalizers. Part (ii) is a uniform feasibility error bound over
model-generated compatible systems: row violations control distance to the
latent feasible completion. Part (iii) requires sampling variation in
certificate directions that can determine rejection; its second clause covers
nearby directions when zero is the unique population maximizer.

\begin{theorem}[Uniform candidatewise coverage]
\label{thm:uniform_candidate_test_validity}
Suppose Assumptions~\ref{ass:comparison_validity} and
\ref{ass:admissible_rule}--\ref{ass:finite_admissible_representation},
\ref{ass:sampling}, and
\ref{ass:regular_candidate_systems} hold uniformly over \(\mathcal P^U\).
Then the set in
\eqref{eq:uniform_inverted_cs} satisfies, for every
\(\alpha\in(0,1/2)\),
\[
\limsup_{n\to\infty}
\sup_{P\in\mathcal P^U}
\sup_{\tau^c\in\Theta_\tau(P;L)\cap\mathcal T}
\PP_P\{\tau^c\notin\mathcal I_n\}
\le \alpha .
\]
\end{theorem}

\begin{proof}
See \hyperref[proof:uniform_candidate_test_validity]{Appendix~\ref*{app:estimated_system_regularity}}.
\end{proof}

The theorem is uniform over DGPs and compatible candidates, but its event is
candidatewise: it controls exclusion of one compatible structural value at a
time. It does not assert simultaneous inclusion of the complete identified
set. Inverting the candidate tests therefore requires no multiplicity
adjustment for this candidatewise guarantee.

Additional theoretical results are collected in the
Appendix.\footnote{Proposition~\ref{prop:separated_incompatibility_power}
establishes power against separated incompatibility. Appendix
Remark~\ref{rem:conservative_calibration} gives a conservative calibration that
does not require part (iii) of
Assumption~\ref{ass:regular_candidate_systems}.}

The next section specifies the finite mean input and admissible restriction for
LAWA and reports inversions for fixed values of \(L\) over a prespecified
sensitivity grid.

\section{Application: The Legal Arizona Workers Act}\label{sec:empirical_illustration}

I apply the framework to the Legal Arizona Workers Act (LAWA), studied by
\citet{bohn2014lawa} as a policy that may have reduced Arizona's likely
unauthorized immigrant population. LAWA mandated E-Verify for new hires and
took effect on January 1, 2008. Because legal status is not observed in
standard household surveys, \citet{bohn2014lawa} use population shares based
on Hispanic noncitizens as proxies for the policy target. Their synthetic
control estimates indicate a sizable decline in Arizona after LAWA.

Migration makes this a natural setting in which comparison states may also be
affected. A
policy that changes the attractiveness of Arizona can redirect prospective
migrants, induce current residents to move, and alter migration into or out of
the United States. \citet{ellis2014migrationresponse} study interstate
out-migration from Arizona, while \citet{amuedodorantes2019interstate} examine
the destination states of Arizona out-migrants. \citet{orrenius2016everify}
document diversion of new arrivals under E-Verify laws. I use this evidence
to motivate possible spillovers, but I do
not require the Arizona effect to be offset exactly within the observed
comparison pool. Instead, I place an interpretable upper bound on total
absolute spillovers across comparison states.

\subsection{Data and Benchmark Construction}

The empirical implementation uses Version 13.0 of the IPUMS Current
Population Survey (CPS) \citep{flood2025ipumscps} to construct annual
state-level series for 1998--2009. I aggregate Basic Monthly records to
state-year means using the final person-level weight \texttt{WTFINL}.
The outcome is the share of
Hispanic noncitizens in the civilian noninstitutional population represented
in the CPS, one of the measures used by \citet{bohn2014lawa}. The sample
contains \(1{,}216{,}961\) household
trajectories linked by the CPS household identifier (CPSID). For state \(k\)
and year \(t\), the empirical cell mean is
\[
\hat\mu_t^k
=
\frac{\sum_{i\in(k,t)} \omega_i Y_i}{\sum_{i\in(k,t)} \omega_i},
\]
where \(Y_i\) is the outcome indicator and \(\omega_i\) is \texttt{WTFINL}.

\begin{samepage}
Arizona is treated. Following \citet{bohn2014lawa}, I exclude Mississippi,
Rhode Island, South Carolina, and Utah from the donor pool.\footnote{The
excluded states had broadly applied restrictions on the employment of
undocumented immigrants.} The pool contains 46 units: 45 states and the
District of Columbia. Let
\(\mathcal K_{\mathrm{cmp}}\) denote this set. The
comparison domain is its full simplex. I use 1998--2006 as the pre-treatment
window, omit 2007 as a transition year, and define the post-treatment target
from the 2008--2009 average.
Inference reflects household-cluster sampling variation conditional on the
policy timing, donor pool, CPS \texttt{WTFINL} weights, and external population
constants; Appendix~\ref{app:arizona_policy_coding} gives the bootstrap design.
\end{samepage}

Relative to \citet{bohn2014lawa}, the main design change is how I construct the
outcome changes before and after treatment. The eight changes from 1999 through
2006 summarize the pre-treatment period, and the 2008--2009 average relative
to the 2006 level is the post-treatment outcome change.
Appendix~\ref{app:arizona_policy_coding} describes the external population
scaling weights and the alternative outcome for the working-age population
with low education.

Figure~\ref{fig:lawa_raw_outcome_trends} displays the annual cell means before
the identifying restrictions are imposed. It shows the entire donor pool
rather than a selected synthetic series because the framework does not
privilege a single comparison weight. Arizona's post-2007 decline is therefore
a feature of the observed data, not by itself an estimate of the policy effect.

\begin{figure}[!htbp]
\centering
\includegraphics[width=0.94\textwidth]{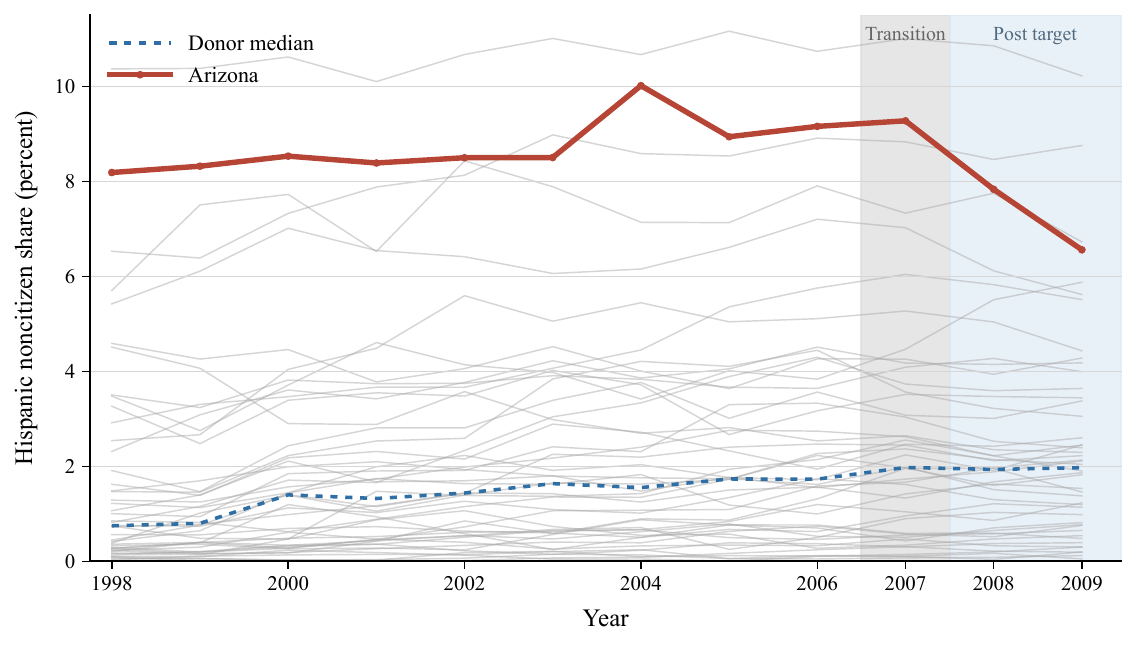}
\caption{Observed LAWA Outcome Paths in Arizona and the Donor Pool}
\label{fig:lawa_raw_outcome_trends}
\begin{minipage}{0.94\textwidth}
\footnotesize
\textit{Notes:} The outcome is the share of Hispanic noncitizens in the
civilian noninstitutional population. The red line is Arizona, the thin gray
lines are the 46 donor units, and the dashed blue line is their unweighted
median across states. All shares use \texttt{WTFINL}.
The gray band marks the omitted 2007
transition year; the blue band marks the 2008--2009 post-treatment target
window. The figure is descriptive: it neither selects a comparison weight nor
imposes the paper's identifying restrictions.
\end{minipage}
\end{figure}

\subsection{Empirical Restrictions}

Comparison validity restricts differences between the treatment effect and
spillovers but does not determine their overall level. I therefore impose a
substantive bound on total absolute spillovers after scaling each state's share
response by population. Because LAWA is an Arizona policy, the total response
across comparison states should not be arbitrarily large relative to the
population potentially affected in Arizona. Specifically, I bound it by a
prespecified multiplier \(\rho>0\) times the size of Arizona's target
population without LAWA. This
restriction does not require spillovers to share a sign or to offset Arizona's
treatment effect.

A natural reference value is \(\rho=1\). A one-person-equivalent response in
one comparison state contributes one unit to the gross-response scale. Using
the 2006 population scaling defined below, \(\rho=1\) therefore permits a total
resident-population-equivalent response as large as Arizona's entire target
population without LAWA.
This benchmark is already broad because it uses the full population of
Hispanic noncitizens rather than only new hires or workers directly covered by
LAWA. The main specification sets \(\rho=2\), allowing twice this benchmark to
accommodate diverted entry, further relocation, and offsetting responses across
comparison states. I report
\(\rho=1\) and \(\rho=4\) as tighter and looser sensitivity specifications.
These values are prespecified calibrations, not estimated parameters.

To put the share responses on a common population scale, let
\(N_k^{\mathrm{ref}}\) denote comparison unit \(k\)'s externally measured 2006
resident population for \(k\in\mathcal K_{\mathrm{cmp}}\), and define
\[
q_k=\frac{N_k^{\mathrm{ref}}}{N_{\mathrm{AZ}}^{\mathrm{ref}}},
\]
where \(N_{\mathrm{AZ}}^{\mathrm{ref}}\) is the corresponding Arizona reference population.
I construct these constants from the U.S. Census Bureau's
\texttt{ST-EST00INT-AGESEX} state intercensal population file
\citep{uscensus2012intercensal}. They are fixed before constructing the CPS
outcome moments and held fixed in the bootstrap. Multiplying a share effect by
its reference population expresses the change on a
resident-population-equivalent scale using the 2006 population base.\footnote{
This is a standardization of the effect, not an assertion that population
denominators after treatment remain equal to their 2006 values. The Census
resident-population base is not identical to the CPS civilian
noninstitutional outcome universe, so the resulting quantities are
resident-population-equivalent responses rather than literal counts of movers
or affected people.}

Let
\(
\mu_T^{\mathrm{AZ}}(\mathbf d^0)
=\mu_T^{\mathrm{AZ}}(\mathbf d^1)-\tau
\)
denote Arizona's no-policy post-treatment share. For a prespecified
\(\rho>0\), the application imposes the gross spillover budget
\begin{equation}\label{eq:empirical_gross_spillover_budget}
\sum_{k\in\mathcal K_{\mathrm{cmp}}}q_k|s_k|
\le
\rho\,\mu_T^{\mathrm{AZ}}(\mathbf d^0).
\end{equation}
In units based on the 2006 populations, this restriction is
\[
\sum_{k\in\mathcal K_{\mathrm{cmp}}}N_k^{\mathrm{ref}}|s_k|
\le
\rho\,N_{\mathrm{AZ}}^{\mathrm{ref}}\mu_T^{\mathrm{AZ}}(\mathbf d^0).
\]
The left-hand side is the total absolute response across comparison units, and
the right-hand side is \(\rho\) times Arizona's target population without LAWA.
Scaling by the population without LAWA rather than the observed population
after LAWA keeps the allowance tied to the population that could have been
affected in the absence of LAWA. This is a scale restriction rather than a conservation equation: it
does not equate Arizona's treatment effect with the net response across
comparison states, and adjustment through international movement or excluded states
need not appear on the left-hand side. At
\(\tau=\tau^c\), the normalized budget on the right-hand side of
\eqref{eq:empirical_gross_spillover_budget} equals
\(\rho\{\mu_T^{\mathrm{AZ}}(\mathbf d^1)-\tau^c\}\). A more negative candidate
implies a larger target population without LAWA and therefore a larger
admissible total response across comparison units. The budget is consequently asymmetric in
\(\tau^c\), which can contribute to different behavior at the lower and upper
endpoints of the compatibility set.

The application also imposes the known support of the share outcome:
\[
0\le \mu_T^{\mathrm{AZ}}(\mathbf d^1)-\tau\le1,
\qquad
0\le \mu_T^k(\mathbf d^1)-s_k\le1,
\quad k\in\mathcal K_{\mathrm{cmp}}.
\]
These inequalities rule out counterfactual shares outside \([0,1]\) and ensure
that the no-policy Arizona stock on the right-hand side of
\eqref{eq:empirical_gross_spillover_budget} is nonnegative.

The gross spillover budget can be written exactly as a finite linear system.
Introduce \(z_k\ge0\) satisfying
\(-z_k\le s_k\le z_k\). Then
\[
\sum_{k\in\mathcal K_{\mathrm{cmp}}}q_k z_k+\rho\tau
\le
\rho\,\mu_T^{\mathrm{AZ}}(\mathbf d^1).
\]
Together with outcome support, these rows define
\(\Omega^{\mathrm{emp}}_\rho(\Pi_0)\) and satisfy
Assumption~\ref{ass:finite_admissible_representation}.\footnote{Bootstrap
draws recompute these post-treatment means.}
The three reported calibrations are nested:
\[
\Omega^{\mathrm{emp}}_1(\Pi_0)
\subseteq
\Omega^{\mathrm{emp}}_2(\Pi_0)
\subseteq
\Omega^{\mathrm{emp}}_4(\Pi_0).
\]

\subsection{Calibration of the Validity Scale Using Observed Periods}
\label{subsec:empirical_l_calibration}

Because \(L\) compares an unobserved gap after treatment with observed fit
before treatment, its numerical scale is difficult to interpret. I therefore
use the placebo indices in
Definition~\ref{def:observed_period_calibration}. The 1998--2006 pre-treatment
window contains eight annual first differences. For each factor dimension
\(d_f\in\{1,2,3\}\), I hold out one difference at a time, estimate the factor
approximation on the other seven, and compute the smallest envelope over the
full simplex that covers the held-out fitted movement. Denote the resulting
envelope by \(\widehat L_{\ell,d_f}^{\mathrm{fac}}\).
Appendix Procedure~\ref{proc:factor_placebo_lbr} gives the calculation.

Let
\[
\widehat B_{d_f}
:=
\max_{\ell=2,\ldots,T_0}
\widehat L_{\ell,d_f}^{\mathrm{fac}}
\]
denote the most demanding held-out movement for factor dimension \(d_f\).
For the main outcome,
\[
(\widehat B_1,\widehat B_2,\widehat B_3)
=
(2.332,\ 2.684,\ 1.526).
\]
The reported path \(L\in[1,3]\) therefore begins below all three
benchmarks and ends above all of them. At \(L=2\), the envelope exceeds the
benchmark based on three factors but not those based on one or two. The
alternative subgroup's corresponding maxima are
\((1.800,\,1.770,\,1.466)\) and serve only as a scale check. Appendix
Table~\ref{tab:arizona_lbr_placebo_full} reports all period-specific indices.

Following \citet{hsu2013calibrating}, this exercise expresses an abstract
sensitivity parameter in units of observed variation. It also follows the idea
in \citet{rambachan2023more} of using movements before treatment to interpret a
departure after treatment. The benchmarks neither select \(L\) nor enter test
construction.

\subsection{Test Inversion Results}

For each \(L\in\{1,1.25,\ldots,3\}\) and \(\rho\in\{1,2,4\}\), I construct a
95 percent inversion set by inverting the candidate compatibility tests over
the full donor simplex. All specifications impose outcome support and use the
homogeneous comparison envelope with a zero floor. The main specification
sets \(\rho=2\); the \(\rho=1\) and \(\rho=4\) specifications vary only the
gross spillover budget.

For reference, using nine pre-intervention years, Table 3 of
\citet{bohn2014lawa} reports a decline of 1.50 percentage points for the share
of all residents who are Hispanic noncitizens, with a one-sided placebo
\(p\)-value of \(0.021\). The inversion sets below show whether zero or the
original estimate is rejected and characterize the other treatment effect
values not rejected under the maintained comparison validity envelopes and
gross spillover budgets.

Figure~\ref{fig:arizona_gm_inversion_paths} reports the lower and upper
boundaries of the 95 percent inversion sets obtained from the compatibility
tests and the adaptive candidate grid. These endpoints numerically approximate
the inversion over all candidate values in
Theorem~\ref{thm:uniform_candidate_test_validity}. In
the audit using a denser grid reported in Appendix~\ref{app:arizona_policy_coding},
every endpoint changes by at most \(1.2\times10^{-4}\).

\begin{figure}[!htbp]
\centering
\includegraphics[width=0.94\textwidth]{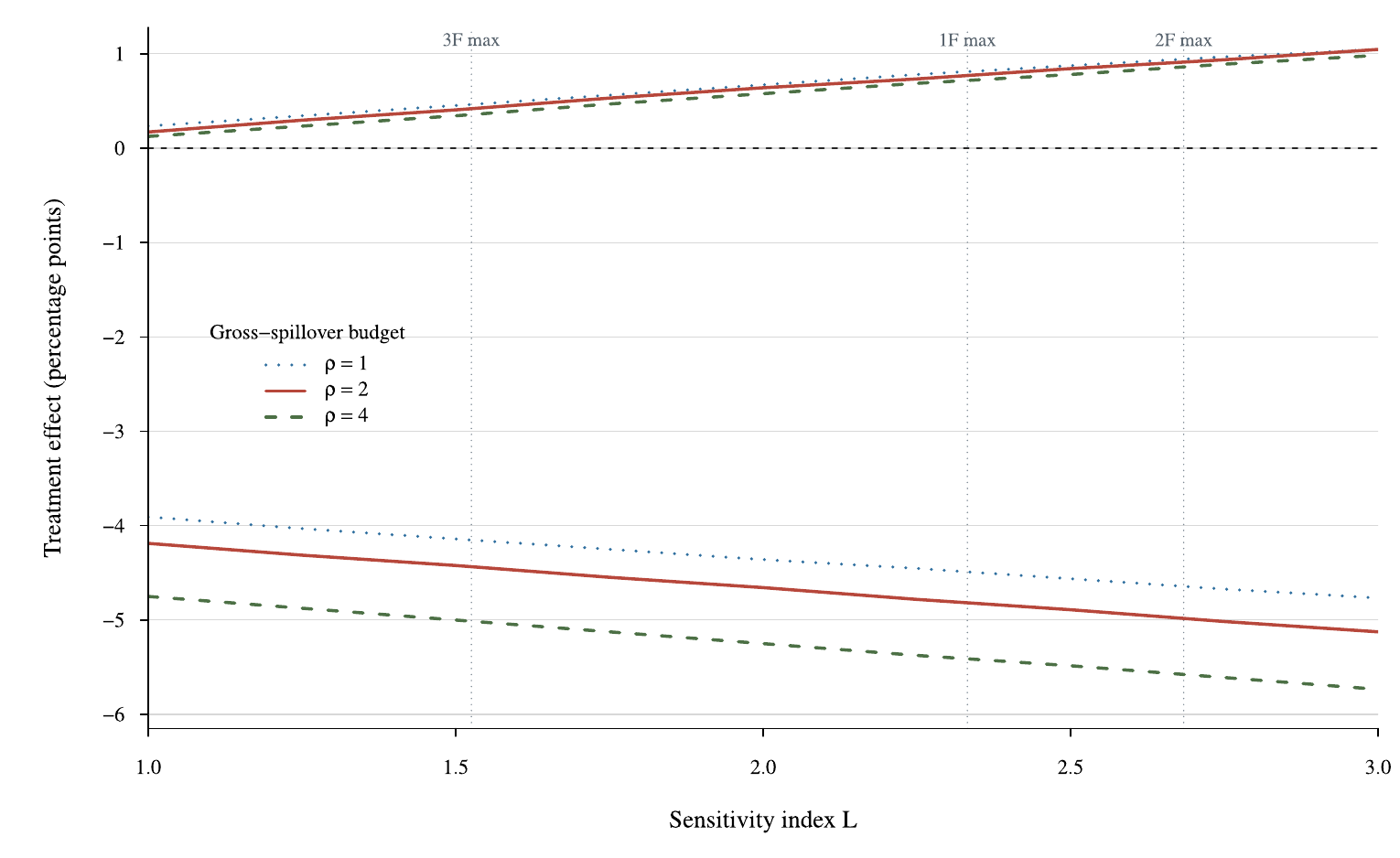}
\caption{Inversion of LAWA Compatibility Tests Across \(L\) and \(\rho\)}
\label{fig:arizona_gm_inversion_paths}
\begin{minipage}{0.94\textwidth}
\footnotesize
\textit{Notes:} Each pair of lines with the same style gives the numerically
approximated lower and upper boundary of the 95 percent inversion set for the
indicated gross spillover budget. Coverage is candidatewise within each
fixed \((L,\rho)\) specification. The paths connect nine separately evaluated
specifications with fixed \(L\) from \(1\) to \(3\) in increments of \(0.25\); they
provide neither simultaneous coverage of the complete identified set nor a
simultaneous band over \(L\) or \(\rho\). The horizontal dashed line marks
zero. The three vertical gray lines mark the maximum factor placebo index for
specifications with one, two, or three factors.
\end{minipage}
\end{figure}

The main \(\rho=2\) inversion set is
\([-4.19,0.17]\) percentage points at \(L=1\),
\([-4.66,0.64]\) at \(L=2\), and
\([-5.13,1.05]\) at \(L=3\). Thus the inversion set contains zero even at the
smallest reported envelope, \(L=1\), which lies below all three factor
placebo benchmarks. Relaxing comparison validity further widens the
range of treatment effects not rejected as compatible.

Allowing a larger gross response primarily extends the negative boundary. At
\(L=2\), the widths of the inversion sets are \(5.03\), \(5.30\), and \(5.83\)
percentage points for \(\rho=1,2,4\), respectively. The upper boundaries are
not ordered in the same way. Population compatibility sets are nested in
\(\rho\), but the separately constructed inversion sets need not be: both the
near-optimal certificate set and its bootstrap critical value depend on the
specification.

Across all reported \((L,\rho)\) specifications, the inversion sets contain
both zero and the decline of 1.50 percentage points reported by
\citet{bohn2014lawa}. The maintained model and sampling uncertainty therefore
do not yield a sign conclusion: neither zero nor a decline of the original
magnitude is rejected. This conclusion is conditional on validity over the full
simplex and the reported gross spillover budgets.

\section{Monte Carlo Evidence}\label{sec:monte_carlo}

The simulation addresses two questions. First, how much do restrictions over
all convex weights narrow the population set when aggregation cancels
pre-treatment discrepancies? Second, how do the tests perform in these
settings? The design holds the empirical post-treatment inputs and admissible
rule fixed while varying only the pattern of pre-treatment discrepancies and
sampling precision.

\subsection{Design and Population Geometry}

The sampling experiment for both designs uses the full simplex comparison
domain, \(L=2\), the outcome support restrictions, and the main
gross spillover budget \(\rho=2\). The population exercise additionally solves
the benchmark using only donor vertices. Both designs preserve each comparison unit's
empirical total absolute pre-treatment discrepancy. The \emph{sign-aligned}
design makes every pre-treatment gap for every donor and period nonnegative.
Convex aggregation then
cannot reduce fit through sign cancellation. The \emph{near-exact} design
places five percent of each state's absolute discrepancy in a common positive
component and chooses donor signs for the remaining 95 percent to minimize the
fit of the convex comparison using population weights. The common component rules out an
exact convex fit. For \(q=(q_k)_{k\in\mathcal K_{\mathrm{cmp}}}\), define the
normalized population weights \(w_q=q/(\mathbbm 1^\top q)\). The ratio
\[
\frac{\sum_{t=2}^{T_0}|C_t(w_q)|}
     {\sum_k w_{q,k}\sum_{t=2}^{T_0}|C_t(e_k)|}
\]
is \(1\) in the sign-aligned design and \(0.0511\) in the near-exact design.
Appendix~\ref{app:simulations} gives the construction.

Table~\ref{tab:mc_population_geometry} reports the deterministic population
exercise before introducing sampling uncertainty. Under sign alignment, the
projections based on donor vertices and the full simplex coincide. Under near-exact fit, the
vertex projection is unchanged because every donor's absolute pre-treatment
discrepancy is held fixed, whereas the projection over the full simplex contracts from
\([-3.365,-1.070]\) to \([-2.130,-1.807]\) percentage points. Its width falls
from \(2.295\) to \(0.324\) percentage points, a reduction of \(85.9\) percent.
Comparing the two weighting domains within each design shows how much interior
weights add when pre-treatment discrepancies cancel.

\begin{table}[!htbp]
\centering
\caption{Population Identification: Donor Vertices and Full Simplex}
\label{tab:mc_population_geometry}
\begin{threeparttable}
\footnotesize
\setlength{\tabcolsep}{7pt}
\begin{tabular}{llccc}
\toprule
Geometry
& Comparison domain
& Lower endpoint
& Upper endpoint
& Width \\
\midrule
Sign-aligned & Donor vertices & \(-3.365\) & \(-1.070\) & \(2.295\) \\
Sign-aligned & Full simplex   & \(-3.365\) & \(-1.070\) & \(2.295\) \\
Near-exact   & Donor vertices & \(-3.365\) & \(-1.070\) & \(2.295\) \\
Near-exact   & Full simplex   & \(-2.130\) & \(-1.807\) & \(0.324\) \\
\bottomrule
\end{tabular}
\begin{tablenotes}[flushleft]
\footnotesize
\item \textit{Notes:} Endpoints and widths are in percentage points. The
donor vertex domain maintains comparison validity only for weights that place
all mass on one donor; the full simplex domain maintains it for every convex
weight. All rows use \(L=2\), \(\rho=2\), outcome support, and the same
post-treatment population inputs.
\end{tablenotes}
\end{threeparttable}
\end{table}

The sampling experiment uses the full simplex model in both geometries. It
follows a finite-dimensional Gaussian experiment calibrated to the
household-cluster variation in the application. The covariance is estimated
from \(399\) cluster multiplier draws. In a cell with precision multiplier
\(m\in\{0.5,1,2\}\), the covariance is divided by \(m\), so larger values of
\(m\) correspond to greater precision. Every cell uses \(500\) Monte Carlo
replications and \(299\) bootstrap draws. Appendix~\ref{app:simulations}
describes the reconstruction of the candidate systems, bootstrap calibration,
and use of common random numbers across cells.

For each cell, I report false exclusion at the lower and upper population
endpoints and the median total width of the inverted acceptance set. Power is
evaluated at candidates one and two percentage points below the lower endpoint
and above the upper endpoint. Fixed absolute distances are used because the
two population projections have very different widths; normalizing the
alternatives by projection width would test economically different deviations.

\subsection{Results}

Table~\ref{tab:mc_main_results} reports the performance in finite samples of
the candidate tests and their inversion for the two population geometries.

\begin{table}[!htbp]
\centering
\caption{Monte Carlo Performance of Candidate Tests}
\label{tab:mc_main_results}
\begin{threeparttable}
\scriptsize
\setlength{\tabcolsep}{2.3pt}
\begin{tabular*}{\textwidth}{@{\extracolsep{\fill}}llccccccc@{}}
\toprule
Geometry
& \(m\)
& \multicolumn{2}{c}{False exclusion}
& \multicolumn{2}{c}{Power: 1 p.p.}
& \multicolumn{2}{c}{Power: 2 p.p.}
& Median width \\
\cmidrule(lr){3-4}\cmidrule(lr){5-6}\cmidrule(lr){7-8}
& & Lower & Upper & Below & Above & Below & Above & (p.p.) \\
\midrule
Sign-aligned & 0.5 & 0.018 & 0.004 & 0.088 & 0.042 & 0.310 & 0.276 & 7.410 \\
Sign-aligned & 1   & 0.024 & 0.006 & 0.202 & 0.100 & 0.632 & 0.548 & 5.851 \\
Sign-aligned & 2   & 0.026 & 0.006 & 0.366 & 0.300 & 0.934 & 0.890 & 4.786 \\
Near-exact   & 0.5 & 0.014 & 0.008 & 0.056 & 0.060 & 0.252 & 0.340 & 5.469 \\
Near-exact   & 1   & 0.014 & 0.006 & 0.120 & 0.164 & 0.520 & 0.656 & 3.922 \\
Near-exact   & 2   & 0.016 & 0.008 & 0.254 & 0.390 & 0.896 & 0.948 & 2.815 \\
\bottomrule
\end{tabular*}
\begin{tablenotes}[flushleft]
\scriptsize
\item \textit{Notes:} The abbreviation p.p. denotes percentage points. Each row uses
\(500\) Monte Carlo replications and
\(299\) bootstrap draws. Lower and Upper are rejection frequencies at the
compatible population endpoints. Below and Above are rejection frequencies at
fixed distances from the corresponding endpoint. Median width is the median
total width of the inversion set, in percentage points. Endpoint tests use
the population widths under the full simplex: \(2.295\) percentage points under sign
alignment and \(0.324\) under near-exact fit. The largest Monte Carlo standard
error in the table is
\(0.022\).
\end{tablenotes}
\end{threeparttable}
\end{table}

Estimated rates of false exclusion range from \(0.004\) to \(0.026\) across the
twelve endpoint experiments. With \(500\) replications per cell, these
estimates provide no evidence of excessive rejection relative to the nominal
\(0.05\) level.

Power rises with both absolute separation and sampling precision. At a
distance of two percentage points, rejection rises from \(0.252\)--\(0.340\) to
\(0.896\)--\(0.948\) in the near-exact design as \(m\) increases from \(0.5\)
to \(2\). The corresponding rise under sign alignment is from
\(0.276\)--\(0.310\) to \(0.890\)--\(0.934\). Alternatives one percentage point
away remain materially harder to distinguish, particularly at the
lower precision levels.

The near-exact geometry also yields narrower inverted sets at every precision
level. Relative to sign alignment, the median width is lower by \(26\), \(33\),
and \(41\) percent at \(m=0.5,1,2\). The population projection falls by \(86\)
percent, so sampling uncertainty attenuates the additional identifying content
of convex comparisons whose discrepancies nearly cancel.

\section{Conclusion}

This paper asks what panel comparisons reveal about treatment effects when
interference to comparison units are unknown. Each comparison relates the
treatment effect to a weighted average of spillovers rather than identifying it
separately. Requiring the same validity rule scaled by fit for every convex donor
weight restricts these relative effects, while restrictions tailored to the
application determine which treatment effect values remain possible. Their
intersection yields the sharp identified set.

Under the finite representation, checking a candidate reduces exactly to
solving a finite linear system, even though validity is required for infinitely
many donor weights. Farkas' alternative provides a certificate when the system
has no solution. Building on the normalized solvability test and bootstrap
calibration of \citet{goff2025inference}, I establish uniform candidatewise
validity for this structured family under model-tailored feasibility and
variance conditions. Inverting the tests controls false exclusion of each
compatible value in large samples. The procedure remains defined when the
identified set is a half-line or the real line, with reported inversion
limited to the compact domain \(\mathcal T\).

The LAWA application combines outcome support with a gross spillover budget
scaled by population. Under the main calibration \(\rho=2\), as well as the
tighter and looser calibrations \(\rho\in\{1,4\}\), the data and maintained
restrictions are consistent with both no effect and the original study's
estimate of a 1.50 percentage point decline. Factor placebos benchmark the
reported values of \(L\) against observed variation without selecting \(L\) or
estimating its latent threshold. The maintained comparison validity envelopes,
gross spillover budgets, and sampling uncertainty therefore do not determine
the sign of the treatment effect.

The simulations isolate when restrictions over all weights matter. Interior
weights do not tighten population identification when donor discrepancies are
sign aligned, but they can sharply narrow the set when convex aggregation
nearly eliminates pre-treatment mismatch. That gain depends on requiring
validity for every donor weight and, in the baseline model, a zero floor at
exact population fit. Sampling uncertainty
attenuates the population gain, while power rises and inversion width falls
with precision.

Panel comparisons can therefore remain informative without an exposure mapping
or advance knowledge of which comparison units are affected. Restrictions
tailored to the application determine what can be learned about the treatment
effect and can be combined with richer models of untreated outcomes.

\clearpage
\begingroup
\raggedright
\bibliographystyle{plainnat}

\endgroup
\clearpage
\appendix
\phantomsection
\label{app:appendix}
\counterwithin{table}{section}
\counterwithin{figure}{section}
\renewcommand{\thetable}{\Alph{section}.\arabic{table}}
\renewcommand{\thefigure}{\Alph{section}.\arabic{figure}}
The appendix first documents the Monte Carlo design and numerical diagnostics,
then gives application and calibration details, and finally collects the
technical proofs.

\section{Monte Carlo Design}\label{app:simulations}
\subsection{Controlled Comparison Geometry}

The experiment starts from the empirical finite-dimensional input for the
main outcome: the pre-treatment donor gaps, post-treatment donor contrasts,
post-treatment levels for the treated and comparison units, and external 2006
Census population weights. It holds every input except the pre-treatment gaps
fixed. The population exercise evaluates two comparison domains: the donor
vertices and the full simplex. The sampling experiment uses the full simplex.
Both impose \(L=2\), outcome support, and
\[
\sum_{k\in\mathcal K_{\mathrm{cmp}}}q_k|s_k|
\le
2\mu_T^1(\mathbf d^0).
\]

Let \(C_t^{\mathrm{emp}}(e_k)\) denote the empirical pre-treatment gap for
donor \(k\) in coordinate \(t\), and use superscripts \(\mathrm{SA}\) and
\(\mathrm{NE}\) for the sign-aligned and near-exact designs. In the
sign-aligned design,
\[
C_\ell^{\mathrm{SA}}(e_k)=|C_\ell^{\mathrm{emp}}(e_k)|,
\qquad \ell=2,\ldots,T_0.
\]
Every convex comparison therefore has discrepancies of the same sign in every
pre-treatment coordinate.

The near-exact design preserves each column's total absolute discrepancy but
creates substantial cancellation. Its first coordinate is
\[
C_2^{\mathrm{NE}}(e_k)
=0.05\sum_{j=2}^{T_0}|C_j^{\mathrm{emp}}(e_k)|.
\]
For \(\ell=3,\ldots,T_0\), set
\[
C_\ell^{\mathrm{NE}}(e_k)
=0.95\,\sigma_k
\left(\sum_{j=2}^{T_0}|C_j^{\mathrm{emp}}(e_k)|\right)
\frac{|C_\ell^{\mathrm{emp}}(e_k)|}
     {\sum_{j=3}^{T_0}|C_j^{\mathrm{emp}}(e_k)|},
\qquad \sigma_k\in\{-1,1\}.
\]
The signs solve the finite binary balancing problem that minimizes
\(
\sum_{\ell=2}^{T_0}|C_\ell^{\mathrm{NE}}(w_q)|
\), where \(w_q=q/(\mathbbm 1^\top q)\), with one sign fixed to remove the
global sign symmetry.
\begin{samepage}
By construction, every comparison unit satisfies
\[
\sum_{\ell=2}^{T_0}|C_\ell^{\mathrm{NE}}(e_k)|
=
\sum_{\ell=2}^{T_0}|C_\ell^{\mathrm{SA}}(e_k)|
=
\sum_{\ell=2}^{T_0}|C_\ell^{\mathrm{emp}}(e_k)|
\]
for every comparison unit. The common positive first coordinate implies the
lower bound
\[
\min_{w\in\Delta^{K-2}}
\frac{\sum_{\ell=2}^{T_0}|C_\ell^{\mathrm{NE}}(w)|}
     {\sum_k w_k\sum_{\ell=2}^{T_0}|C_\ell^{\mathrm{NE}}(e_k)|}
\ge0.05.
\]
\end{samepage}
For the constructed sign pattern, a finite LP over \(w\) attains this lower
bound up to solver tolerance (the computed value differs from \(0.05\) by less
than \(2\times10^{-16}\)). The design therefore approaches but does not reach
exact convex fit. At the normalized population weights, the ratio is
\(0.051058\), compared with \(1\) under sign alignment. The minimum average
absolute pre-treatment discrepancy in the near-exact design is \(0.000201\).

Let superscripts \(\mathrm{DV}\) and \(\mathrm{FS}\) denote the domain of donor
vertices and the full simplex, respectively. Because both designs preserve
every donor's total absolute pre-treatment discrepancy, their systems at the
donor vertices
coincide at \(L=2\) and \(\rho=2\):
\[
\Theta_\tau^{\mathrm{SA},\mathrm{DV}}
=
\Theta_\tau^{\mathrm{NE},\mathrm{DV}}
=[-0.033650,-0.010704].
\]
Solving the population systems over the full simplex gives
\[
\Theta_\tau^{\mathrm{SA},\mathrm{FS}}
=[-0.033650,-0.010704],
\qquad
\Theta_\tau^{\mathrm{NE},\mathrm{FS}}
=[-0.021302,-0.018067].
\]
Thus interior weights do not sharpen the sign-aligned design. In the
near-exact design, they reduce the width from \(0.022946\) to \(0.003235\), or
by \(85.9\) percent.

\subsection{Sampling and Numerical Implementation}

For each geometry \(G\in\{\mathrm{SA},\mathrm{NE}\}\), let \(\Pi_G\) collect
the population gaps, post-treatment contrasts, and outcome levels. A rank-398
covariance factor \(\widehat\Sigma_G^{1/2}\) is calibrated from \(399\)
multiplier draws clustered by household. Its pre-treatment coordinates receive the
same sign transformations as the population gaps, preserving their marginal
variances while matching the controlled geometry.

For precision \(m\in\{0.5,1,2\}\) and replication \(j\), the simulated input
has the form
\[
\widehat\Pi_{G,m}^{(j)}
=
\Pi_G+m^{-1/2}\widehat\Sigma_G^{1/2}Z_j,
\qquad Z_j\sim N(0,I).
\]
The \(299\) bootstrap perturbations use independent standard Gaussian vectors
with the same covariance factor and precision scaling. They determine the row
scales and the perturbations of estimated entries in the bootstrap certificate
objectives. For each candidate, the minimum-norm near-feasible completion and
near-optimal certificate region are computed from the sample system; the
bootstrap draws change only the certificate objective. The implementation uses
\(\gamma_n=\sqrt{\log(300)}\), as in the empirical procedure, and common random
numbers across the six combinations of design and precision.

False exclusion is evaluated at both population endpoints. Power candidates
are \(0.01\) and \(0.02\) below the lower endpoint and above the upper endpoint.
The full inversion begins on \([-0.08,0.04]\), uses a coarse spacing of
\(0.01\), and bisects every acceptance transition to a tolerance of
\(2.5\times10^{-4}\). All accepted components are retained. One near-exact
replication at \(m=0.5\) initially reached the lower numerical boundary and was
recomputed on \([-0.16,0.12]\); no final inversion ends at a numerical boundary.

\section{Implementation Diagnostics}
\label{app:implementation_diagnostics}

The diagnostics in this appendix document the numerical behavior of the
implemented LP systems. They report failed candidate and bootstrap programs,
cases in which the inversion set reaches the numerical treatment effect
support, and whether the accepted grid points form a connected set.

For each empirical \((L,\rho)\) specification, the reported diagnostics record
the number of finite candidate programs, valid bootstrap draws, connected
acceptance components, boundary hits, scales for estimated rows, norms of
near-feasible solutions, and constraint residuals from quadratic programs.
These quantities are recomputed with the gross spillover budget and its
auxiliary rows for absolute values rather than carried over from a different
admissible rule.

All \(3{,}000\) Monte Carlo inversions succeed. Their adaptive searches evaluate
\(110{,}332\) candidate values, and the power exercise at fixed distances adds
\(12{,}000\) candidate tests. Every test uses all \(299\) bootstrap draws and no
bootstrap LP fails. Every inversion is nonempty and connected, and no final
inversion reaches the numerical search boundary.

\section{Application Details and Calibration}\label{app:arizona_policy_coding}

The empirical application follows \citet{bohn2014lawa} in using annualized CPS
data for 1998--2009, dropping 2007 as a transition year, and excluding
Mississippi, Rhode Island, South Carolina, and Utah from the comparison pool.
The treated state is Arizona, and the outcome is the share of Hispanic
noncitizens in the state population represented by the CPS. The empirical
candidate systems use pre-treatment first differences over 1998--2006 and the
2008--2009 average relative to the 2006 level as the post-treatment target.

The empirical bootstrap takes the complete trajectory identified by a CPS
household identifier (CPSID) as the sampling cluster. Before inference, the
person-month records are aggregated into weighted numerator and denominator
contributions for each household and year. Each draw assigns every CPSID one
independent exponential multiplier with mean and variance one and uses that
multiplier for all household members and survey months in the trajectory.
CPSIDs are stratified by state and cohort of entry into the sample, and the
multipliers are normalized to have mean one within each stratum. The draw then
recomputes the annual state ratios, first differences, and inputs to the
candidate system. The economic observations entering identification remain
annual state means; clustering allows repeated CPS interviews to determine
their sampling covariance rather than their independent sample size. Here
\(n\) is the number of distinct CPSID trajectories in the analysis sample.
This procedure targets superpopulation variation at the household-cluster
level, conditional on the published CPS \texttt{WTFINL} weights. It is not
intended to reproduce design-based CPS variance estimation under the full multistage
sampling design \citep{censusbls2000cpsdesign}.

The bootstrap holds fixed the policy timing, comparison pool, rule defining the
admissible set, and external constants. The population scaling weights
\(q_k=N_k^{\mathrm{ref}}/N_{\mathrm{AZ}}^{\mathrm{ref}}\) in the gross
spillover budget are external design constants from the 2006 resident
population estimates in the Census \texttt{ST-EST00INT-AGESEX} file
\citep{uscensus2012intercensal}. They are fixed before constructing the CPS
outcome moments and held fixed in the bootstrap. The main specification fixes
\(\rho=2\), while the robustness
specifications separately fix \(\rho=1\) and \(\rho=4\). For a candidate
\(\tau^c\), the budget right-hand side is
\(\rho\{\hat\mu_T^{\mathrm{AZ}}(\mathbf d^1)-\tau^c\}\). The treated post-treatment
mean and all means in the \([0,1]\) support restrictions are recomputed in each draw and enter
the corresponding sample or bootstrap candidate system as estimated inputs to
the admissible set.
The computation uses the exact finite system in
Theorem~\ref{prop:candidate_tau_linear_system}, with the same fixed row
order for the sample and every bootstrap draw. Generated rows use their
bootstrap standard deviation estimates directly, while fixed rows receive
scale zero. With \(B\) bootstrap draws, the implementation sets the three
common tuning values to \(\gamma_n=\sqrt{\log(B+1)}\). The corresponding
asymptotic implementation takes \(B=B_n\to\infty\) with
\(\log B_n=o(n)\), which gives \(\gamma_n\to\infty\) and
\(\gamma_n/\sqrt n\to0\).
The empirical inputs were prepared in R 4.4.1 \citep{rcoreteam2024r}, while the
candidate systems, empirical inversions, and Monte Carlo experiments were
computed in Python 3.14.3. The linear programs and the strictly convex
quadratic program selecting the minimum-norm nuisance vector were solved with
HiGHS 1.15.1 through its Python interface \citep{huangfu2018parallelizing}.
SciPy's SLSQP implementation \citep{virtanen2020scipy} was used for the same
quadratic program only if HiGHS did not return an optimal solution. The
empirical bootstrap uses random seed \(12345\), and the Monte Carlo uses seed
\(20260721\). For each
candidate, the bootstrap certificate program has a common feasible region
across draws, so the implementation constructs that program once and changes
only its objective. The linear programs use presolve and warm starts.

The empirical inversion uses \(B=399\), the numerical candidate domain
\([-0.08,0.04]\), an initial grid spacing of \(0.01\), and bisection of every
detected acceptance transition to a tolerance of \(2.5\times10^{-4}\).
The initial grid also includes zero, the two benchmarks that assume clean
comparisons, and the feasibility endpoints and midpoint of the sample system
as numerical search anchors. Conditional quantiles use the median-unbiased
type-8 rule of \citet{hyndman1996sample} and a \(10^{-6}\) numerical
tie-breaking increment. The candidate domain and search
anchors do not enter the solvability system for a fixed candidate as additional
restrictions. Across the 27 reported
\((L,\rho)\) specifications, every candidate test uses all 399 bootstrap
draws, no bootstrap LP fails, every accepted set has one connected component,
and no accepted set reaches the boundary of the reporting domain. A denser
audit of all 27 specifications, using the same 399 bootstrap draws, an initial
grid spacing of \(0.0025\), and boundary tolerance \(5\times10^{-5}\), changes
each reported endpoint by at most \(1.2\times10^{-4}\) and finds no additional
accepted component.
The reported critical values are Monte Carlo approximations to the ideal
conditional quantiles; Theorem~\ref{thm:uniform_candidate_test_validity} does
not include the simulation error from holding \(B=399\) fixed.

The empirical inference and the Monte Carlo use the same form of admissible
rule: outcome support and the gross spillover budget scaled by population. The
Monte Carlo fixes the main calibration \((L,\rho)=(2,2)\) and varies only
pre-treatment comparison geometry and sampling precision.

\paragraph{Alternative outcome definition.}
\citet{bohn2014lawa} also study the share among residents ages 15--45 who are
Hispanic, noncitizens, and have at most a high school education. For the scale
check below, I construct this subgroup share from the CPS. The subgroup is used
only to assess whether the scale of \(L\) from observed periods is sensitive to the
proxy definition; it is not analyzed as a second application of the treatment
effect analysis.

Table~\ref{tab:arizona_lbr_placebo_full} reports the complete benchmarks used
to interpret the comparison validity scale. For the main outcome, the maximum
in each factor dimension gives one vertical reference line in
Figure~\ref{fig:arizona_gm_inversion_paths}. These maxima are \(2.332\),
\(2.684\), and \(1.526\). The maxima for the alternative subgroup, \(1.800\),
\(1.770\), and \(1.466\), are reported only as a scale check.

\begin{procedure}[Factor placebo calibration of \(L\) using observed periods]
\label{proc:factor_placebo_lbr}
Let \((\Delta\hat\mu_\ell^k)\) be the matrix of estimated
pre-treatment first differences. For each outcome definition and factor
dimension \(d_f\in\{1,2,3\}\), compute the benchmark as follows.
\begin{enumerate}[label=(\roman*),leftmargin=*]
\item Treat each pre-treatment first difference \(\ell\) once as a pseudo-post
movement. Let \(X_{kh}=\Delta\hat\mu_h^k\) for \(h\ne\ell\). Using equal weights
on all unit--time cells, remove the sample grand mean
\(\bar X_{-\ell}\) and the unit and time effects from \(X\), and take the
rank-\(d_f\) truncated singular value decomposition (SVD) of the resulting
double-centered matrix. Splitting its singular values symmetrically between
the left and right singular vectors yields loadings
\(\widehat\varphi_k\) and factors \(\widehat f_h\) in
\[
\Delta\hat\mu_h^k
=
\widehat a_k+\widehat b_h
+\widehat\varphi_k^\top\widehat f_h+\widehat e_{kh},
\qquad h\ne\ell,\quad \dim(\widehat f_h)=d_f .
\]
Here \(\widehat a_k\) and \(\widehat b_h\) are the fitted unit and time
effects, and \(\widehat e_{kh}\) is the fitted residual.
\item For \(h\ne\ell\), use the fitted additive effects and rank-\(d_f\)
component. For the held-out cross-section, regress
\(\Delta\hat\mu_\ell^k-\widehat a_k-\bar X_{-\ell}\) on a constant
and \(\widehat\varphi_k\), again with equal unit weights. Use its fitted
intercept and factor coefficients to obtain
\(\widetilde{\Delta\mu}_{\ell,-\ell}^{k,(d_f)}\). Denote all fitted differences
by \(\widetilde{\Delta\mu}_{h,-\ell}^{k,(d_f)}\).
\item For each comparison weight \(w\), form
\[
\widetilde C_{h,-\ell}^{(d_f)}(w)
:=
\widetilde{\Delta\mu}_{h,-\ell}^{1,(d_f)}
-
\sum_{k=2}^K w_k\,\widetilde{\Delta\mu}_{h,-\ell}^{k,(d_f)} .
\]
The resulting index is
\[
\widehat L_{\ell,d_f}^{\mathrm{fac}}
:=
\inf\left\{
L\ge0:
\left|\widetilde C_{\ell,-\ell}^{(d_f)}(w)\right|
\le
L\,
\frac{1}{T_0-2}
\sum_{h\ne \ell}
\left|\widetilde C_{h,-\ell}^{(d_f)}(w)\right|
\text{ for all } w\in\Delta^{K-2}
\right\}.
\]
\end{enumerate}
Each entry in Table~\ref{tab:arizona_lbr_placebo_full} is therefore the minimum
value of \(L\) that covers one pseudo-post movement over the full simplex for a
prespecified factor dimension.
\end{procedure}

The computation first fits two-way additive effects by least squares, then
applies a truncated SVD and the held-out cross-sectional projection just
described. Only the fitted values enter the calibration; the implementation
gives equal weight to every unit--time cell in every holdout, and SVD sign
changes leave the fitted values and projection onto the loading space
unchanged.

The procedure is the sample analogue of
\(L_\ell^{\mathrm{pl}}\) in
Definition~\ref{def:observed_period_calibration}, based on smoothed factors.
It is informative when
low-dimensional pre-treatment movements provide a useful scale for latent
post-treatment gaps. Example~\ref{ex:factor_comparison_validity} links both
objects to the sup-norm coefficients that represent one factor movement with
the others. The main text inference is computed separately at
each displayed fixed value of \(L\); the calibration indices are reference
markers, not estimated specification parameters.

\begin{table}[!htbp]
\centering
\caption{Factor Placebo Benchmarks for \(L\) from Observed Periods}
\label{tab:arizona_lbr_placebo_full}
\begin{threeparttable}
\scriptsize
\setlength{\tabcolsep}{4.4pt}
\begin{tabular}{@{}lcccccc@{}}
\toprule
Pseudo-post first difference
& \multicolumn{3}{c}{Main outcome}
& \multicolumn{3}{c}{Alternative subgroup} \\
\cmidrule(lr){2-4}\cmidrule(l){5-7}
& 1 factor
& 2 factors
& 3 factors
& 1 factor
& 2 factors
& 3 factors \\
\midrule
1999--1998 & \(2.332\) & \(2.684\) & \(1.526\) & \(1.547\) & \(1.239\) & \(0.908\) \\
2000--1999 & \(1.348\) & \(1.022\) & \(0.847\) & \(1.371\) & \(1.074\) & \(0.676\) \\
2001--2000 & \(1.775\) & \(1.448\) & \(1.520\) & \(1.104\) & \(1.747\) & \(0.998\) \\
2002--2001 & \(1.951\) & \(1.981\) & \(1.297\) & \(1.266\) & \(1.307\) & \(0.968\) \\
2003--2002 & \(1.410\) & \(1.182\) & \(1.008\) & \(1.673\) & \(1.511\) & \(1.466\) \\
2004--2003 & \(1.265\) & \(1.334\) & \(1.276\) & \(1.297\) & \(1.060\) & \(1.088\) \\
2005--2004 & \(1.175\) & \(1.220\) & \(1.234\) & \(1.800\) & \(1.770\) & \(1.200\) \\
2006--2005 & \(1.218\) & \(1.362\) & \(0.928\) & \(1.076\) & \(0.790\) & \(0.527\) \\
\midrule
Maximum & \(2.332\) & \(2.684\) & \(1.526\) & \(1.800\) & \(1.770\) & \(1.466\) \\
\bottomrule
\end{tabular}
\begin{tablenotes}[flushleft]
\footnotesize
\item \textit{Notes:} Each row treats the listed pre-treatment first difference
as a pseudo-post moment. For each factor dimension, the approximation is
estimated on the other seven differences, and all comparison moments use the
corresponding fitted values. Each entry is the minimum value of \(L\) for which
the envelope over the full simplex covers the pseudo-post factor movement. The
maximum row gives the three main text figure markers for the main outcome; the alternative
outcome maxima are reported here as a scale check.
\end{tablenotes}
\end{threeparttable}
\end{table}

\subsection{Primitive Sampling Verification}
\label{app:primitive_sampling_verification}

The application uses ratios of weighted cluster totals. The following result
records primitive conditions under which this sampling design satisfies the
high-level Gaussian and bootstrap approximation used for inference.

\begin{lemma}[Clustered weighted-ratio sampling]
\label{lem:cluster_ratio_sampling_sufficient}
Consider a fixed finite collection of aggregate-unit and period cells. Let
sampling cluster \(h\) contribute a weighted numerator \(U_{hkt}\) and
denominator \(V_{hkt}\) to cell \((k,t)\), and form
\[
\hat\mu_t^k=\frac{\sum_{h=1}^n U_{hkt}}{\sum_{h=1}^n V_{hkt}}.
\]
Suppose the following conditions hold uniformly over \(P\in\mathcal P^U\):
\begin{enumerate}[label=(\roman*)]
\item The \(n\) clusters are independent and are identically distributed
either overall or within each of \(S<\infty\) fixed strata. If \(n_s\) is the
nonrandom number of clusters in stratum \(s\), then, for constants
\(\pi_s\) satisfying \(\min_{s\le S}\pi_s>0\),
\[
\max_{s\le S}\sqrt n\,|n_s/n-\pi_s|\to0.
\]
Dependence among records within a cluster is unrestricted.
\item The vector \(R_h\) stacking all \((U_{hkt},V_{hkt})\) has a uniformly
bounded \(2+\delta\) moment for some \(\delta>0\), and every population
denominator is bounded away from zero. The population input \(\Pi(P)\) lies in
a fixed compact set, and \(\hat\Pi_n\) is a fixed continuously differentiable
function of the cell ratios and any exactly observed inputs, with derivative
uniformly bounded on a common neighborhood of their population values.
\item The bootstrap assigns independent nonnegative weights to clusters,
independently of the sample, with mean and variance one and a bounded
\(2+\delta\) moment. It recomputes all cell ratios and \(\hat\Pi_n\). Under
stratified sampling, the weights are normalized to have sample mean one within
each stratum.
\end{enumerate}
Then Assumption~\ref{ass:sampling} holds.
\end{lemma}

\begin{proof}
For record-level data, the cluster contributions can be written as
\[
U_{hkt}:=\sum_{i\in\mathcal C_h}\omega_iY_i
\mathbbm 1\{G_i=k,T_i=t\},
\qquad
V_{hkt}:=\sum_{i\in\mathcal C_h}\omega_i
\mathbbm 1\{G_i=k,T_i=t\},
\]
where \(\mathcal C_h\) contains every record in cluster \(h\), while
\((G_i,T_i)\), \(Y_i\), and \(\omega_i\) are record \(i\)'s unit and period
labels, outcome, and sampling weight, and \(\mathbbm 1\{\cdot\}\) is the
indicator function. After incorporating the fixed limiting stratum shares,
the corresponding population ratio is the ratio of the average expectations
of \(U_{hkt}\) and \(V_{hkt}\).

Write \(\pi_{s,n}=n_s/n\) and let \(\bar R_{s,n}\) be the sample mean of
\(R_h\) within stratum \(s\). The vector of aggregate cluster contributions is
a sum of the form \(\sum_s\pi_{s,n}\bar R_{s,n}\), and the stated rate permits
replacement of \(\pi_{s,n}\) by \(\pi_s\) at \(o(n^{-1/2})\). Because the
dimension and number of strata are fixed, independence, the nonvanishing
stratum shares, and the uniform \(2+\delta\) moment bound give a uniform
stratified triangular-array Gaussian approximation for the vector of
sample-average cluster contributions. The cluster multiplier bootstrap
consistently reproduces that approximation and its covariance. Within-stratum
normalization only replaces each multiplier by its deviation from the stratum
average, up to a factor converging uniformly to one, and therefore preserves
the stratified first-order law.

The lower bound on the population denominators makes the map from average
cluster contributions to the cell ratios continuously differentiable on a
common neighborhood. The additional map from those ratios to \(\Pi\) has the
same property by assumption. The uniform delta method, applied to the sample
and recomputed bootstrap vectors, gives the Gaussian, bootstrap, and covariance
approximations in Assumption~\ref{ass:sampling}
\citep{vandervaart1998asymptotic,vandervaartwellner1996weak,
kosorok2008introduction}.
\end{proof}

\section{Technical Proofs}
\label{app:technical_proofs}

\subsection{Identification and Candidate Representation}
\label{app:candidate_system_representation_proofs}

For \(w\in\Delta^{K-2}\), write
\[
C_t(w)
:=
\Delta\mu_t^1(\mathbf d^0)
-\sum_{k=2}^K w_k\Delta\mu_t^k(\mathbf d^0),
\]
and
\[
C_T^{\mathrm{obs}}(w)
:=
\Delta\mu_T^1(\mathbf d^1)
-\sum_{k=2}^K w_k\Delta\mu_T^k(\mathbf d^1).
\]
Define
\[
a_L(w):=\frac{L}{T_0-1}\sum_{t=2}^{T_0}|C_t(w)|.
\]

\begin{example}[Factor-span derivation and observed-period counterpart]
\label{ex:factor_comparison_validity}
Suppose untreated first differences satisfy
\[
\Delta\mu_t^k(\mathbf d^0)=\alpha_t+f_t^\top h_k,
\qquad t=2,\ldots,T_0,T,\quad k=1,\ldots,K,
\]
where \(\alpha_t\) is a common increment and \(f_t,h_k\in\mathbb R^r\). For
\(w\in\Delta^{K-2}\), let
\[
\delta_h(w):=h_1-\sum_{k=2}^K w_k h_k.
\]
Convex weights eliminate the common increment, so the untreated comparison
gap in period \(t\) is \(f_t^\top\delta_h(w)\). Let
\(F_{\mathrm{pre}}=[f_2,\ldots,f_{T_0}]\). If
\(f_T=F_{\mathrm{pre}}\gamma\), then, simultaneously for every \(w\),
\[
|f_T^\top\delta_h(w)|
\le
\|\gamma\|_\infty
\sum_{t=2}^{T_0}|f_t^\top\delta_h(w)|.
\]
Hence Assumption~\ref{ass:comparison_validity} holds whenever
\(L\ge(T_0-1)\|\gamma\|_\infty\). A sufficient condition for the span
requirement is \(T_0-1\ge r\) and
\(\operatorname{rank}(F_{\mathrm{pre}})=r\). Under that condition, exact
pre-treatment fit implies a zero post-treatment latent gap. The smallest
\(\|\gamma\|_\infty\) over representations of \(f_T\) gives the tightest
factor-model bound.

The same argument gives the observed-period counterpart in
Definition~\ref{def:observed_period_calibration}. Hold out
\(\ell\in\{2,\ldots,T_0\}\) and let
\[
F_{-\ell}
:=
[f_2,\ldots,f_{\ell-1},f_{\ell+1},\ldots,f_{T_0}].
\]
If \(f_\ell=F_{-\ell}\gamma_\ell\), then
\[
|f_\ell^\top\delta_h(w)|
\le
\|\gamma_\ell\|_\infty
\sum_{\substack{t=2\\t\ne\ell}}^{T_0}
|f_t^\top\delta_h(w)|
\quad\text{for every }w,
\]
and hence
\[
L_\ell^{\mathrm{pl}}
\le
(T_0-2)
\inf_{\gamma_\ell:\,f_\ell=F_{-\ell}\gamma_\ell}
\|\gamma_\ell\|_\infty.
\]
Thus the latent threshold and the observed-period indices measure factor
movements in the same sup-norm representation. When the comparison-loading
directions contain a neighborhood of zero in \(\mathbb R^r\), norm duality
makes the last bound exact.
\end{example}

\paragraph{Admissible restriction templates.}
The three anchoring roles described after
Proposition~\ref{prop:exposure_anchoring} admit simple finite representations.
Outcome support for the treated unit gives
\[
\mu_T^1(\mathbf d^1)-\overline y_1
\le\tau\le
\mu_T^1(\mathbf d^1)-\underline y_1
\]
when its no-policy mean lies in
\([\underline y_1,\overline y_1]\). An aggregate band
\[
\underline b
\le a_\tau\tau+a_s^\top s\le
\overline b
\]
has two-sided anchoring power when
\(a_\tau+a_s^\top\mathbbm 1_{K-1}\ne0\). Coordinate bounds
\(|s_k|\le\overline s_k\) and gross budgets
\[
\sum_{k=2}^K q_k|s_k|\le\overline S,
\qquad q_k>0,
\]
also block both common-shift rays and have standard linear lifts. A single
sign or aggregate inequality may block only one ray. By contrast, ordering
restrictions such as \(s_k\ge s_\ell\), and network or exposure rows whose
coefficients sum to zero, are shift invariant and do not anchor the common
level by themselves. Stacking any prespecified collection of these rows gives
the finite representation in
Assumption~\ref{ass:finite_admissible_representation}.

\begin{lemma}[Exact comparison equivalence]
\label{lem:exact_comparison_equivalence}
Let \(W\) stack a finite collection of convex comparison weights as its rows.
If \((\tau,s)\) satisfies the associated exact comparison system, then its
complete observational equivalence class is
\[
(\tau,s)
+
\left\{
\left(\delta,\delta\mathbbm 1_{K-1}+v\right):
\delta\in\mathbb R,\ v\in\ker(W)
\right\},
\]
where \(\ker(W):=\{v\in\mathbb R^{K-1}:Wv=0\}\).
\end{lemma}

\begin{proof}
Two pairs \((\tau,s)\) and \((\widetilde\tau,\widetilde s)\) generate the
same exact comparison contrasts if and only if
\[
W\left\{
(\widetilde\tau-\tau)\mathbbm 1_{K-1}
-(\widetilde s-s)
\right\}=0.
\]
Set \(\delta=\widetilde\tau-\tau\) and
\(v=(\widetilde s-s)-\delta\mathbbm 1_{K-1}\). The display is equivalent to
\(v\in\ker(W)\), giving the stated class. Conversely, every pair in that
class leaves \(W(\tau\mathbbm 1_{K-1}-s)\) unchanged.
\end{proof}

\begin{proposition}[Interior comparison restriction gain]
\label{prop:comparison_multiplicity}
Let \(e_k\) denote the simplex vertex placing all weight on donor \(k\).
For every \(w\in\Delta^{K-2}\),
\[
a_L(w)\le\sum_{k=2}^K w_k a_L(e_k).
\]
The comparison inequality indexed directly by \(w\) is therefore weakly
tighter than the inequality obtained by convexly aggregating the donor-vertex
rows. If the display is strict, the interior row is not implied by the vertex
rows and strictly reduces the comparison-validity region. Whether it strictly
reduces the treatment effect projection additionally depends on
\(\Omega(\Pi_0)\).
\end{proposition}

\begin{proof}
For every pre-treatment period,
\[
C_t(w)
=
\sum_{k=2}^K w_k
\left\{
\Delta\mu_t^1(\mathbf d^0)-\Delta\mu_t^k(\mathbf d^0)
\right\}.
\]
The triangle inequality, summed over periods, proves the displayed bound.
Moreover,
\[
C_T^{\mathrm{obs}}(w)-\tau+w^\top s
=
\sum_{k=2}^K w_k
\{C_T^{\mathrm{obs}}(e_k)-\tau+s_k\}.
\]
Thus the vertex inequalities imply the interior inequality with right-hand
side \(\sum_k w_k a_L(e_k)\). If the fit inequality is strict, choose
\(s_k=a_L(e_k)-C_T^{\mathrm{obs}}(e_k)+\tau\). The positive-side vertex rows
then hold at equality, and the negative-side vertex rows also hold because
\(a_L(e_k)\ge0\). Their convex aggregate violates the tighter interior row.
Hence that row strictly reduces the comparison-validity region.
\end{proof}

\phantomsection
\label{proof:treatment_effect_projection}
\begin{proof}[Proof of Proposition~\ref{prop:treatment_effect_projection}]
Fix \(\tau^c\in\mathbb R\). Suppose first that
\(\mathsf{Comp}(\Pi_0,\tau^c)=1\). For the corresponding maintained-model
completion, set \(x=\tau^c\mathbbm 1_{K-1}-s\). Comparison validity gives
\(x\in\mathcal X_L(\Pi_0)\), and admissibility gives
\((\tau^c,\tau^c\mathbbm 1_{K-1}-x)\in\Omega(\Pi_0)\). This proves the forward direction
of \eqref{eq:sharp_completion_general}.

Conversely, suppose that some \(x\in\mathcal X_L(\Pi_0)\) satisfies
\((\tau^c,\tau^c\mathbbm 1_{K-1}-x)\in\Omega(\Pi_0)\), and set
\(s=\tau^c\mathbbm 1_{K-1}-x\). Define counterfactual post-treatment differences by
\[
\widetilde{\Delta\mu}_T^1(\mathbf d^0)
=
\Delta\mu_T^1(\mathbf d^1)-\tau^c,
\qquad
\widetilde{\Delta\mu}_T^k(\mathbf d^0)
=
\Delta\mu_T^k(\mathbf d^1)-s_k,
\]
and set
\[
\widetilde\mu_T^k(\mathbf d^0)
=
\mu_{T_0}^k(\mathbf d^0)
+\widetilde{\Delta\mu}_T^k(\mathbf d^0),
\qquad k=1,\ldots,K.
\]
Retain the observed values as the potential means under \(\mathbf d^1\).
Because \(x\in\mathcal X_L(\Pi_0)\), these means satisfy every comparison
validity inequality. Membership in \(\Omega(\Pi_0)\) ensures that all
prespecified admissible restrictions, including outcome support when imposed,
also hold. Potential means under assignments other than \(\mathbf d^0\) and
\(\mathbf d^1\) can be completed arbitrarily because they do not enter the
model. This constructs a maintained-model completion with effect \(\tau^c\), proving
the reverse direction and the stated attainment result for
\eqref{eq:sharp_completion_general}.
The identified-set formula \eqref{eq:sharp_anchor_union} follows by collecting
all compatible candidate values.
\end{proof}

\phantomsection
\label{proof:finite_comparison_bridge}
\begin{proof}[Proof of Lemma~\ref{lem:finite_comparison_bridge}]
Let \(r_k=C_T^{\mathrm{obs}}(e_k)-x_k\). For the positive signed inequality,
\[
\frac{L}{T_0-1}\sum_{t=2}^{T_0}|C_t(w)|
=
\max_{\|v\|_\infty\le L/(T_0-1)}
\sum_{t=2}^{T_0}v_tC_t(w).
\]
The simplex and the box are compact and convex, and the criterion is bilinear.
The minimax theorem of \citet{sion1958general} therefore gives
\[
\begin{aligned}
\sup_{w\in\Delta^{K-2}}
\left\{
\sum_{k=2}^K w_kr_k
-\tfrac{L}{T_0-1}\sum_{t=2}^{T_0}|C_t(w)|
\right\}
&=
\sup_w\min_{\|v\|_\infty\le L/(T_0-1)}
\sum_{k=2}^K w_k
\left\{r_k-\sum_{t=2}^{T_0}v_tC_t(e_k)\right\}\\
&=
\min_{\|v\|_\infty\le L/(T_0-1)}
\max_{2\le k\le K}
\left\{r_k-\sum_{t=2}^{T_0}v_tC_t(e_k)\right\}.
\end{aligned}
\]
The positive signed inequalities hold for every \(w\) if and only if the last
display is nonpositive, which is equivalent to
\[
r_k\le\sum_{t=2}^{T_0}v_t^+C_t(e_k),
\qquad k=2,\ldots,K,
\]
for some
\[
v^+\in
\left[-\frac{L}{T_0-1},\frac{L}{T_0-1}\right]^{T_0-1}.
\]
Applying the same argument to \(-r\) gives
\[
\sum_{t=2}^{T_0}v_t^-C_t(e_k)\le r_k,
\qquad k=2,\ldots,K,
\]
for some
\[
v^-\in
\left[-\frac{L}{T_0-1},\frac{L}{T_0-1}\right]^{T_0-1}.
\]
Combining the two statements proves the
lemma.
\end{proof}

\phantomsection
\label{proof:candidate_tau_linear_system}
\begin{proof}[Proof of Theorem~\ref{prop:candidate_tau_linear_system}]
Fix \(\tau^c\in\mathbb R\). By
Proposition~\ref{prop:treatment_effect_projection}, compatibility is equivalent
to the existence of an \(x\) satisfying \eqref{eq:sharp_completion_general}.
Lemma~\ref{lem:finite_comparison_bridge} replaces
\(x\in\mathcal X_L(\Pi_0)\) exactly by the existence of
\[
v^-,v^+
\in
\left[-\frac{L}{T_0-1},\frac{L}{T_0-1}\right]^{T_0-1}
\]
such that
\[
\sum_{t=2}^{T_0}v_t^-C_t(e_k)
\le C_T^{\mathrm{obs}}(e_k)-x_k
\le\sum_{t=2}^{T_0}v_t^+C_t(e_k),
\qquad k=2,\ldots,K.
\]
These are the comparison and box rows in
\eqref{eq:bridge_candidate_master_system}.
The finite representation of the admissible rule supplies \(u\), and
substituting
\(s=\tau^c\mathbbm 1_{K-1}-x\) gives
\[
A_\tau(\Pi_0)\tau^c
+A_s(\Pi_0)s+A_u u
=
-A_s(\Pi_0)x+A_u u+g\tau^c,
\]
where the last equality uses
\eqref{eq:fixed_common_shift_loading}. Rearranging gives the last block of
\eqref{eq:bridge_candidate_master_system}. Reversing these steps shows that
every solution of the finite system supplies an \(x\) satisfying
\eqref{eq:sharp_completion_general}. Therefore the general compatibility
condition is equivalent to
\(\mathcal F(\Pi_0,\tau^c)\ne\varnothing\), proving
\eqref{eq:sharp_completion_finite_solvability}.
\end{proof}

\phantomsection
\label{proof:exposure_anchoring}
\begin{proof}[Proof of Proposition~\ref{prop:exposure_anchoring}]
Let \(\mathcal C_L(\Pi_0)\) denote the set of pairs \((\tau,s)\) that satisfy
the full-simplex comparison restrictions. Lemma
\ref{lem:finite_comparison_bridge} gives this set a finite polyhedral
representation, and Assumption~\ref{ass:finite_admissible_representation}
makes \(\Omega(\Pi_0)\) a polyhedron. Their intersection is nonempty under the
maintained model.

A vector \(h=(h_\tau,h_s)\) is a recession direction of
\(\mathcal C_L(\Pi_0)\) if, for every feasible \((\tau,s)\) and every
\(a\ge0\), the pair
\((\tau+a h_\tau,s+a h_s)\) remains in \(\mathcal C_L(\Pi_0)\). At the
simplex vertex \(e_k\), comparison validity places the relative effect
\(\tau-s_k\) in a finite interval. Along direction \(h\), this relative
effect changes at rate \(h_\tau-h_{s,k}\). Remaining in that interval for
every \(a\ge0\) therefore requires
\(h_{s,k}=h_\tau\) for every donor, or
\(h_s=h_\tau\mathbbm 1_{K-1}\). Conversely, every direction of this form
leaves all relative effects unchanged and is therefore a recession direction.
Thus
\[
\operatorname{rec}\{\mathcal C_L(\Pi_0)\}
=
\operatorname{span}\{\iota\}.
\]

For the nonempty intersection of the two polyhedra,
\[
\operatorname{rec}\{\mathcal C_L(\Pi_0)\cap\Omega(\Pi_0)\}
=
\operatorname{span}\{\iota\}
\cap\operatorname{rec}\{\Omega(\Pi_0)\}.
\]
The treatment effect is bounded above if and only if this cone contains no
direction with \(h_\tau>0\), which is equivalent to
\(\iota\notin\operatorname{rec}\{\Omega(\Pi_0)\}\). The lower-bound result
follows by reversing the sign. Both bounds are finite if and only if the
displayed recession cone is \(\{0\}\), proving the final claim.
\end{proof}

\begin{remark}[A fit-independent additive floor]
\label{rem:additive_floor_extension}
A relaxation adds a prespecified floor
\(\delta_{\mathrm f}\ge0\) to the comparison allowance:
\[
\left|C_T^{\mathrm{obs}}(w)-w^\top x\right|
\le
\delta_{\mathrm f}
+\frac{L}{T_0-1}\sum_{t=2}^{T_0}|C_t(w)|.
\]
For the full simplex, the same minimax argument as in
Lemma~\ref{lem:finite_comparison_bridge} replaces
\eqref{eq:finite_comparison_bridge} by
\[
\sum_{t=2}^{T_0}v_t^-C_t(e_k)-\delta_{\mathrm f}
\le C_T^{\mathrm{obs}}(e_k)-x_k
\le
\sum_{t=2}^{T_0}v_t^+C_t(e_k)+\delta_{\mathrm f},
\qquad k=2,\ldots,K,
\]
with the same box constraints on \(v^-\) and \(v^+\). Hence this floor also
preserves a finite candidate representation. Holding \(L\), the donor pool,
fit score, and admissible rule fixed, increasing \(\delta_{\mathrm f}\) weakly
enlarges the population compatibility set.
The uniform inference result in the main text is stated and verified for the
prespecified rows imposing validity over the full simplex and a zero floor.
\end{remark}

\begin{lemma}[Bounded candidate representations]
\label{lem:bounded_feasible_nuisance}
Suppose Assumptions~\ref{ass:comparison_validity} and
\ref{ass:admissible_rule}--\ref{ass:finite_admissible_representation} hold.
Fix compact sets
\(\mathcal K_\Pi\) in the input domain and \(\mathcal T\subset\mathbb R\), and
hold \(L<\infty\) and the lift of the admissible set fixed. Whenever
\(\mathcal F(\Pi,\tau)\ne\varnothing\), let \(\eta^0(\Pi,\tau)\) be its
minimum-norm element. Then
\[
\sup_{\substack{(\Pi,\tau)\in\mathcal K_\Pi\times\mathcal T:\\
\mathcal F(\Pi,\tau)\ne\varnothing}}
\|\eta^0(\Pi,\tau)\|<\infty.
\]
\end{lemma}

\begin{proof}
Because the finitely many \(C_T^{\mathrm{obs}}(e_k)\) and \(C_t(e_k)\) are
affine in \(\Pi\), compactness of \(\mathcal K_\Pi\) bounds them uniformly.
Because the finite bridge includes every donor vertex \(e_k\), feasibility of
its comparison rows implies
\[
|C_T^{\mathrm{obs}}(e_k)-x_k|
\le
\frac{L}{T_0-1}
\sum_{t=2}^{T_0}|C_t(e_k)|.
\]
It therefore bounds every coordinate of any compatible vector of relative
effects \(x\). The box constraints in
\eqref{eq:bridge_candidate_master_system} directly give
\[
\|v^-\|_\infty\le \frac{L}{T_0-1},
\qquad
\|v^+\|_\infty\le \frac{L}{T_0-1}.
\]
Compactness of \(\mathcal T\) consequently also bounds the implied spillover
vector \(s=\tau\mathbbm 1_{K-1}-x\).

Conditional on bounded \((\Pi,\tau,x)\), the remaining rows defining the
admissible set form a feasible linear system in \(u\) with fixed coefficient
matrix \(A_u\).
Affinity and compactness bound \(A_s(\Pi)\) and \(a(\Pi)\), so its right-hand
side ranges over a compact set. Hoffman's error bound for this fixed auxiliary
matrix \citep{hoffman1952approximate} implies
that the minimum norm feasible selection of \(u\) is uniformly bounded over
all nonempty systems in that right-hand-side set. Combining this selection
with the bounded \(x,v^-\), and \(v^+\) coordinates gives a feasible nuisance
vector whose norm is bounded by a common constant.

For each compatible pair, \(\mathcal F(\Pi,\tau)\) is a nonempty closed convex
polyhedron. The strictly convex objective \(\|\eta\|^2\) therefore has a unique
minimizer, whose norm cannot exceed that of the bounded feasible vector just
constructed. This proves the uniform bound.
\end{proof}
\subsection{Estimated-System Verification}
\label{app:estimated_system_regularity}

This subsection fixes the row representation used by the test and verifies the
conditions underlying the estimated-coefficient procedure of
\citet{goff2025inference} for the systems generated by the identification
analysis. The finite part of
Theorem~\ref{prop:candidate_tau_linear_system} and
Lemma~\ref{lem:bounded_feasible_nuisance} already provide fixed dimension,
candidate-invariant input loadings, and bounded feasible nuisance
representatives. Assumption~\ref{ass:sampling}, with primitive conditions in
Lemma~\ref{lem:cluster_ratio_sampling_sufficient}, supplies the Gaussian and
bootstrap approximation.
Assumption~\ref{ass:regular_candidate_systems}(ii) localizes feasibility
stability to the compatible right-hand sides generated by the model, while
part (iii) localizes the zero-optimum variance requirement to structurally
generated, near-boundary extreme certificates. The results below establish
scale consistency, stability of the minimum-norm completion, and approximation
of the relevant dual certificates. The final proof assembles these pieces and
directly verifies the comparison and quantile-transfer steps affected by the
model-specific conditions.

For compactness, this subsection writes \(\tau\) for the candidate denoted
\(\tau^c\) in the main text. Comparison validity is imposed over the full
simplex as in Assumption~\ref{ass:comparison_validity}, the envelope value
\(L<\infty\) is fixed, and candidate values range over the fixed compact set
\(\mathcal T\).

\subsubsection{Candidate system and bootstrap objects}

\begin{definition}[Prespecified bridge candidate representation]
\label{def:prespecified_rows}
Use exactly the inequality rows displayed in
\eqref{eq:bridge_candidate_master_system}, with
\(\eta=(x^\top,u^\top,(v^-)^\top,(v^+)^\top)^\top\). Any exact admissible
equality has already been encoded by the two prespecified opposite rows in
Assumption~\ref{ass:finite_admissible_representation}. Retain the displayed
nuisance coordinates and row order, and write the resulting system as
\[
b(\Pi,\tau)-A(\Pi)\eta\le0.
\]
The same symbolic construction is used for the population, sample, and
bootstrap systems. It retains the bridge coordinates and box rows in
\eqref{eq:bridge_candidate_master_system} and preserves the additive
decomposition in \eqref{eq:bridge_candidate_affine_decomposition}. The
candidate reporting domain \(\mathcal T\) is not added as a system row.
Auxiliary-variable units are those in
Assumption~\ref{ass:finite_admissible_representation}, and no redundant rows
are added or removed after the representation is fixed.
\end{definition}

The finite row representation is affine in \(\Pi\), and the candidate has a
fixed loading on the right-hand side. Hence there is a fixed matrix
\(M_\Pi\) such that, for every \(\Pi,\Pi'\), and \(\tau\),
\begin{equation}\label{eq:bridge_candidate_affine_decomposition}
\operatorname{vec}
\left[
A(\Pi')-A(\Pi)\quad
b(\Pi',\tau)-b(\Pi,\tau)
\right]
=M_\Pi(\Pi'-\Pi).
\end{equation}
Let \(\xi(P,\tau)\), \(\hat\xi_n(\tau)\), and
\(\hat\xi_n^*(\tau)\) denote the vectorizations of
\([A(\Pi)\ b(\Pi,\tau)]\) evaluated at \(\Pi(P)\), \(\hat\Pi_n\), and
\(\hat\Pi_n^*\), respectively. Assumption~\ref{ass:sampling} therefore
implies, for every \(\tau\in\mathcal T\),
\begin{equation}\label{eq:candidate_common_perturbation}
\begin{aligned}
\sqrt n\{\hat\xi_n(\tau)-\xi(P,\tau)\}&=M_\Pi Z_{n,P},\\
\sqrt n\{\hat\xi_n^*(\tau)-\hat\xi_n(\tau)\}&=M_\Pi Z_n^*.
\end{aligned}
\end{equation}
Both perturbations are constant in \(\tau\).

This construction preserves candidate feasibility. An entry is estimated when
its affine formula has a nonzero loading on the stochastic input \(\Pi\), and
a row is estimated when it contains at least one estimated entry. Every other
entry and row is fixed. This classification is common to the population,
sample, and bootstrap systems. Positive row rescaling and row permutation are
innocuous when the row scales are transformed in the same way. General
nuisance reparameterizations, alternative lifts, and redundant-row additions
are not included in this equivalence because they can alter the minimum-norm
calibration.

Use the population objects \(\mathsf S(P)\), \(\mathcal D(P)\),
\(Q(P,\tau)\), \(\eta^0(P,\tau)\), and \(\Gamma(P,\tau)\) defined in
Section~\ref{subsec:uniform_validity}. Let \(J_e\) denote the fixed set of
rows containing at least one estimated entry, and write \(\lambda_{J_e}\)
for the corresponding coordinates of a multiplier. The population optimal
set is
\[
\Lambda_0(P,\tau)
:=
\arg\max_{\lambda\in\mathcal D(P)}
b(\Pi(P),\tau)^\top\lambda.
\]
The sample matrix \(\widehat{\mathsf S}_n\) uses bootstrap estimates of the
row scales and is common across \(\tau\) by
\eqref{eq:bridge_candidate_affine_decomposition}.

Define \(\hat{\mathcal D}_n\) by sample substitution.

Fix \(\gamma_n\to\infty\) with \(\gamma_n/\sqrt n\to0\) and define
\[
\hat r_n(\tau):=\hat Q_n(\tau)+\gamma_n n^{-1/2}.
\]
This uses the admissible specialization
\(\kappa_{1n}=\kappa_{2n}=\kappa_{3n}=\gamma_n\) of the three tuning sequences
in \citet[equations (3.9)--(3.10)]{goff2025inference}.
The minimum norm near-feasible nuisance solution is
\begin{equation}\label{eq:uniform_eta_hat}
\begin{aligned}
\hat\eta_n(\tau)
\in
\arg\min_{\eta}\|\eta\|^2
\quad\text{subject to}\quad
&\hat b_n(\tau)-\hat A_n\eta
\le
\hat r_n(\tau)\widehat{\mathsf S}_n\mathbbm 1.
\end{aligned}
\end{equation}
The corresponding near-optimal certificate set is
\begin{equation}\label{eq:uniform_near_dual_set}
\begin{aligned}
\widehat\Lambda_n^{\mathrm{near}}(\tau):=\bigl\{\lambda:\ &\lambda\ge0,\\
&-\gamma_n n^{-1/2}\mathbbm 1
\le \hat A_n^\top\lambda
\le \gamma_n n^{-1/2}\mathbbm 1,\\
&\mathbbm 1^\top\widehat{\mathsf S}_n\lambda\le1,\\
&\hat b_n(\tau)^\top\lambda
\ge \hat Q_n(\tau)-\gamma_n n^{-1/2}\bigr\}.
\end{aligned}
\end{equation}
Using the stacked system vectors in
\eqref{eq:candidate_common_perturbation}, define the candidate-invariant
bootstrap perturbation
\[
\zeta_n^*
:=\sqrt n\{\hat\xi_n^*(\tau)-\hat\xi_n(\tau)\}.
\]
Write its matrix and right-hand-side blocks as
\(\zeta_{n,\mathrm{mat}}^*\) and \(\zeta_{n,\mathrm{rhs}}^*\). The bootstrap statistic is
\begin{equation}\label{eq:uniform_bootstrap_stat}
Q_n^*(\tau)
:=
\max_{\lambda\in\widehat\Lambda_n^{\mathrm{near}}(\tau)}
\left\langle
\zeta_{n,\mathrm{rhs}}^*-
\zeta_{n,\mathrm{mat}}^*\hat\eta_n(\tau),\lambda
\right\rangle,
\end{equation}
and \(\hat c_n(\tau,1-\alpha)\) is its conditional
\((1-\alpha)\)-quantile.

Thus the system perturbation itself is computed once, but its evaluation at
the candidate-specific nuisance solution and over the candidate-specific
near-optimal certificate set remains indexed by \(\tau\). Accordingly, the
critical value remains candidate-specific. A finite-grid inversion directly
covers null candidates on that grid. Continuum coverage requires continuum
inversion or an explicit enlargement for grid and numerical error.

\subsubsection{Scale normalization and primal completion}

\begin{lemma}[Estimated-entry scale consistency]
\label{lem:row_scale_consistency}
Under Assumption~\ref{ass:sampling} and part (i) of
Assumption~\ref{ass:regular_candidate_systems}, every
estimated entry scale lies in \([\underline d,\bar d]\), fixed entries have
scale zero, and these scales do not depend on \(\tau\). The bootstrap estimator
of the entry scales is uniformly consistent. The induced row normalizer
\(\widehat{\mathsf S}_n\) is therefore uniformly consistent for
\(\mathsf S(P)\).
\end{lemma}

\begin{proof}
By Assumption~\ref{ass:finite_admissible_representation} and
\eqref{eq:bridge_candidate_affine_decomposition}, each coordinate of
\(\xi(P,\tau)\) is affine in \(\Pi(P)\), with a loading that is independent of
\(\tau\). The bridge rows in
\eqref{eq:bridge_candidate_master_system} and the affine representation of
\(\Omega(\Pi)\) identify these fixed loadings. Fixed coordinates have zero
loading by construction.
Part (i) of Assumption~\ref{ass:regular_candidate_systems} gives the uniform
lower bound for estimated-entry scales. If \(m_j\) is the fixed loading of
entry \(j\) on \(\Pi\), its population variance is
\(m_j^\top\Sigma_\Pi(P)m_j\), and its bootstrap variance estimator is
\(m_j^\top\operatorname{Var}^*(Z_n^*)m_j\). The uniform covariance bound in
Assumption~\ref{ass:sampling}, together with the fixed finite collection of
loadings, gives the common upper bound \(\bar d\); the same assumption gives
uniform convergence of the bootstrap variances. Taking the maximum over the
fixed number of entry scales in each row gives uniform consistency of
\(\widehat{\mathsf S}_n\).
\end{proof}

\begin{lemma}[Uniform behavior of the near-feasible completion]
\label{lem:near_feasible_completion_behavior}
Suppose Assumptions~\ref{ass:comparison_validity} and
\ref{ass:admissible_rule}--\ref{ass:finite_admissible_representation},
Assumption~\ref{ass:sampling}, and part (i) of
Assumption~\ref{ass:regular_candidate_systems} hold. The program in
\eqref{eq:uniform_eta_hat} is nonempty whenever
\(\hat Q_n(\tau)<\infty\). Uniformly over \((P,\tau)\in\mathcal Q_0\),
\[
\hat Q_n(\tau)=O_p(n^{-1/2}),
\qquad
\hat\eta_n(\tau)=O_p(1).
\]
If, in addition, part (ii) of
Assumption~\ref{ass:regular_candidate_systems} holds, then
\[
\|\hat\eta_n(\tau)-\eta^0(P,\tau)\|
=O_p\!\left\{\left(\frac{\gamma_n}{\sqrt n}\right)^{1/2}\right\}
=o_p(1).
\]
\end{lemma}

\begin{proof}
\emph{Existence and boundedness.}
By LP duality, \(\hat Q_n(\tau)\) is the value of
\[
\begin{aligned}
\inf_{\eta,\ r\ge0}\quad & r\\
\text{subject to}\quad
&\hat b_n(\tau)-\hat A_n\eta
\le r\,\widehat{\mathsf S}_n\mathbbm 1.
\end{aligned}
\]
Hence \(r=\hat Q_n(\tau)+\gamma_n n^{-1/2}\) gives the nonempty constraint
set in \eqref{eq:uniform_eta_hat} whenever \(\hat Q_n(\tau)<\infty\).

Let \(\eta^0(P,\tau)\) be the population minimum-norm feasible solution.
Lemma~\ref{lem:bounded_feasible_nuisance} bounds it uniformly,
\eqref{eq:candidate_common_perturbation} makes every estimated row
perturbation \(O_p(n^{-1/2})\), and
Lemma~\ref{lem:row_scale_consistency} consistently estimates row scales that
are bounded away from zero on estimated rows. Fixed rows do not change.
Because \(\gamma_n\to\infty\), \(\eta^0(P,\tau)\) belongs to the relaxed
sample set with probability approaching one uniformly. Minimum-norm
optimality then gives
\(\|\hat\eta_n(\tau)\|\le\|\eta^0(P,\tau)\|\) on that event. The smallest
normalized relaxation needed to accommodate the same population solution is
\(O_p(n^{-1/2})\), which also gives
\(\hat Q_n(\tau)=O_p(n^{-1/2})\).

\emph{Consistency under part (ii).}
Abbreviate
\[
(A,b,\mathcal F)
:=
\bigl(A(\Pi(P)),b(\Pi(P),\tau),\mathcal F(\Pi(P),\tau)\bigr).
\]
For any perturbed \((\widetilde A,\widetilde b)\) and
\(\widetilde\eta\), the decomposition
\[
b-A\widetilde\eta
=
(\widetilde b-\widetilde A\widetilde\eta)
+(b-\widetilde b)
+(\widetilde A-A)\widetilde\eta
\]
and Assumption~\ref{ass:regular_candidate_systems}(ii) imply
\[
\operatorname{dist}\{\widetilde\eta,\mathcal F\}
\le \bar H\bigl(
\|[\widetilde b-\widetilde A\widetilde\eta]_+\|
+\|\widetilde b-b\|
+\|\widetilde A-A\|\,\|\widetilde\eta\|
\bigr).
\]
The constraints defining \(\hat\eta_n(\tau)\), its uniform bound, scale
consistency, and \(\gamma_n/\sqrt n\to0\) give
\[
\|[\hat b_n(\tau)-\hat A_n\hat\eta_n(\tau)]_+\|
=O_p(\gamma_n/\sqrt n).
\]
Together with \eqref{eq:candidate_common_perturbation}, the preceding error
bound yields
\[
d_n(P,\tau)
:=
\operatorname{dist}\{\hat\eta_n(\tau),\mathcal F(\Pi(P),\tau)\}
=O_p(\gamma_n/\sqrt n)
\]
uniformly on \(\mathcal Q_0\). Let \(\widetilde\eta_n(P,\tau)\) be the
Euclidean projection of \(\hat\eta_n(\tau)\) onto this population feasible
set. Then
\[
\|\widetilde\eta_n-\hat\eta_n\|=d_n,
\qquad
\|\widetilde\eta_n\|\le\|\eta^0\|+d_n.
\]
Because \(\eta^0\) is the projection of zero onto the same closed convex set,
\[
\|\widetilde\eta_n-\eta^0\|^2
\le
\|\widetilde\eta_n\|^2-\|\eta^0\|^2
=O_p(\gamma_n/\sqrt n).
\]
The triangle inequality proves the stated rate.
\end{proof}

\subsubsection{Near-optimal dual certificates}

\begin{lemma}[Directed approximation of perturbed certificates]
\label{lem:directed_near_dual_inclusion}
Suppose Assumptions~\ref{ass:comparison_validity} and
\ref{ass:admissible_rule}--\ref{ass:finite_admissible_representation},
Assumption~\ref{ass:sampling}, and part (i) of
Assumption~\ref{ass:regular_candidate_systems} hold.
Let \(\operatorname{pr}_e\lambda=\lambda_{J_e}\), and for two sets define the
directed distance
\[
\overrightarrow d(E,F):=\sup_{z\in E}\inf_{\widetilde z\in F}
\|z-\widetilde z\|.
\]
For a perturbation
\(\zeta=(\zeta_A,\zeta_b)\) supported on estimated entries and a diagonal
normalizer \(\widetilde{\mathsf S}\), let
\(\widetilde\Lambda_{0,n}(P,\tau;\zeta,\widetilde{\mathsf S})\) be the
maximizers of
\[
\max_{\lambda\ge0}
\{b(\Pi(P),\tau)+n^{-1/2}\zeta_b\}^\top\lambda
\]
subject to
\[
\{A(\Pi(P))+n^{-1/2}\zeta_A\}^\top\lambda=0,
\qquad
\mathbbm 1^\top\widetilde{\mathsf S}\lambda\le1.
\]
Call \(\widetilde{\mathsf S}\) admissible when it is zero off \(J_e\) and
its diagonal entries on \(J_e\) lie in
\([\underline d/2,2\bar d]\).
There are deterministic sequences \(M_n\to\infty\) and
\(a_n\downarrow0\), chosen so that
\[
M_n=o(\gamma_n),\qquad
M_na_n=o(1),\qquad
M_n\left(\frac{\gamma_n}{\sqrt n}\right)^{1/2}=o(1),
\]
for which, uniformly over \((P,\tau)\in\mathcal Q_0\),
\[
\sup_{\substack{\|\zeta\|\le M_n,\;
\|\widetilde{\mathsf S}-\mathsf S(P)\|\le a_n\\
\widetilde{\mathsf S}\ \mathrm{admissible}}}
\overrightarrow d\!\left(
\operatorname{pr}_e\widetilde\Lambda_{0,n}
(P,\tau;\zeta,\widetilde{\mathsf S}),
\operatorname{pr}_e\widehat\Lambda_n^{\mathrm{near}}(\tau)
\right)
=O_p(a_n)=o_p(1).
\]
\end{lemma}

\begin{proof}
Uniform consistency of \(\widehat{\mathsf S}_n\) permits a deterministic
sequence \(a_n\downarrow0\) for which the scale event below has probability
approaching one uniformly. Because
\(\gamma_n\to\infty\), \(a_n^{-1}\to\infty\), and
\((\sqrt n/\gamma_n)^{1/2}\to\infty\), choose \(M_n\to\infty\) sufficiently
slowly that
\[
M_n=o\!\left(
\min\left\{\gamma_n,a_n^{-1},
\left(\frac{\sqrt n}{\gamma_n}\right)^{1/2}\right\}
\right).
\]
Uniform tightness of the sample perturbation then implies that, with
probability approaching one uniformly,
\[
\sqrt n\|(\hat A_n-A(\Pi(P)),
\hat b_n(\tau)-b(\Pi(P),\tau))\|\le M_n,
\qquad
\|\widehat{\mathsf S}_n-\mathsf S(P)\|\le a_n.
\]
This choice satisfies all three rate restrictions in the lemma.
Scale normalization and part (i) of
Assumption~\ref{ass:regular_candidate_systems} give a common bound on
\(\|\lambda_{J_e}\|_1\) for every multiplier in any of the normalized programs
in the lemma. Fixed-row entries are unperturbed.
Any recession ray of a perturbed normalized dual program is therefore
supported on fixed rows. Population feasibility makes its objective
nonpositive, so each perturbed program has a finite, attained maximum and
the displayed maximizer set is nonempty.

Uniformly over the displayed perturbation ball, the sample equality residual
of a perturbed optimizer \(\lambda\) satisfies
\[
\|\hat A_n^\top\lambda\|_\infty
=O(M_n/\sqrt n)
=o(\gamma_n/\sqrt n).
\]
Let \(b_\zeta=b(\Pi(P),\tau)+n^{-1/2}\zeta_b\). Because zero is feasible in
the perturbed dual program, optimality gives \(b_\zeta^\top\lambda\ge0\).
Only estimated rows differ between \(b_\zeta\) and \(\hat b_n(\tau)\), so the
estimated-row multiplier bound implies
\[
\hat b_n(\tau)^\top\lambda
\ge -C M_n/\sqrt n
\]
for a common finite \(C\), on the same event.

Because \(\lambda\) is normalized by \(\widetilde{\mathsf S}\) and
\(\|\widetilde{\mathsf S}-\widehat{\mathsf S}_n\|\le2a_n\), define
\[
\rho_n(\lambda)
:=
\max\{1,\mathbbm 1^\top\widehat{\mathsf S}_n\lambda\},
\qquad
\widetilde\lambda:=\lambda/\rho_n(\lambda).
\]
The estimated-row multiplier bound gives
\(1\le\rho_n(\lambda)\le1+Ca_n\), so
\(\widetilde\lambda\) satisfies the sample normalization and
\(\|\operatorname{pr}_e(\widetilde\lambda-\lambda)\|=O(a_n)\).
The rescaling cannot worsen the preceding lower bound: a negative objective
moves toward zero, while a nonnegative objective remains nonnegative. Since
\(\hat Q_n(\tau)=O_p(n^{-1/2})\) and \(M_n=o(\gamma_n)\), uniformly with
probability approaching one,
\[
\hat b_n(\tau)^\top\widetilde\lambda
\ge
\hat Q_n(\tau)-\gamma_n/\sqrt n.
\]
The equality residual is still within the near-set tolerance, so
\(\widetilde\lambda\in
\widehat\Lambda_n^{\mathrm{near}}(\tau)\). Its estimated-row coordinates are
within \(O(a_n)\) of those of \(\lambda\), uniformly over the perturbation
ball, which proves the directed-distance claim.
\end{proof}

\subsubsection{Proof of uniform candidatewise coverage}

\phantomsection
\label{proof:uniform_candidate_test_validity}
\label{app:main_uniform_theorem_proof}
\begin{proof}[Proof of Theorem~\ref{thm:uniform_candidate_test_validity}]
Use the null index class \(\mathcal Q_0\) defined in
Section~\ref{subsec:uniform_validity}. Formally, one may regard \((P,\tau)\)
as the law \(P\) augmented by the deterministic label \(\tau\); probabilities
and bootstrap laws remain those under \(P\). For
\(\theta=(P,\tau)\in\mathcal Q_0\), use the prespecified
inequality system
\[
b_\theta-A_\theta\eta\le0,
\qquad
(A_\theta,b_\theta)=\bigl(A(\Pi(P)),b(\Pi(P),\tau)\bigr).
\]
The statistic and calibration in
\eqref{eq:sample_uniform_infeasibility} and
\eqref{eq:uniform_eta_hat}--\eqref{eq:uniform_bootstrap_stat} correspond to the
all-inequality procedure in \citet[Theorem~4.19]{goff2025inference}. The proof
follows their one-sided comparison and quantile-transfer argument while
adapting the two steps affected by the present structure: the uniform
feasibility error bound controls the near-feasible completion, and the
certificate-variance condition controls structurally reachable certificate
directions that can determine rejection.

\medskip\noindent\emph{Step 1: stochastic inputs and normalization.}
The finite system in
Theorem~\ref{prop:candidate_tau_linear_system} fixes the numbers of nuisance
coordinates, rows, and estimated entries. Assumption~\ref{ass:sampling} and
\eqref{eq:candidate_common_perturbation} give the uniform Gaussian
approximation of
\(\operatorname{vec}[\hat A_n\ \hat b_n(\tau)]\) and a
consistent bootstrap approximation. Lemma~\ref{lem:row_scale_consistency}
gives the entrywise scale bounds and consistent normalizer. These are the
finite-dimensional input conditions used in the cited theorem.

Recall that \(J_e\) denotes the fixed set of estimated rows. On an event whose
probability tends to one uniformly, scale consistency and the lower bound in
part (i) of Assumption~\ref{ass:regular_candidate_systems} imply
\begin{equation}
\sup_{\substack{\lambda\ge0:\
\mathbbm 1^\top\widehat{\mathsf S}_n\lambda\le1}}
\sum_{j\in J_e}\lambda_j\le C
\label{eq:proof_estimated_multiplier_bound}
\end{equation}
for a common finite constant \(C\). Multipliers on fixed rows need not be
bounded, but their matrix and right-hand-side perturbations are identically
zero. Consequently, every stochastic dual objective below depends only on the
uniformly bounded estimated-row part of \(\lambda\).

\medskip\noindent\emph{Step 2: the null-local completion.}
Lemma~\ref{lem:bounded_feasible_nuisance} gives
\[
\sup_{\theta\in\mathcal Q_0}
\|\eta^0(P,\tau)\|<\infty.
\]
Lemma~\ref{lem:near_feasible_completion_behavior} shows that the sample
minimum norm relaxation is well defined on the null with probability
approaching one uniformly and, uniformly on \(\mathcal Q_0\), gives
\[
\begin{aligned}
\hat Q_n(\tau)&=O_p(n^{-1/2}),
&\hat\eta_n(\tau)&=O_p(1),\\
\|\hat\eta_n(\tau)-\eta^0(P,\tau)\|
&=O_p\!\left\{\left(\frac{\gamma_n}{\sqrt n}\right)^{1/2}\right\}.
\end{aligned}
\]
The lemma obtains the last rate by applying the uniform feasibility error bound
in Assumption~\ref{ass:regular_candidate_systems}(ii) and then using
minimum-norm projection.

The subsequent stochastic comparison uses this rate as follows. For the
sequences in Lemma~\ref{lem:directed_near_dual_inclusion},
let \(\mathcal Z_{A,n}\) contain the matrix perturbations supported on the
prespecified estimated entries with norm at most \(M_n\).
\eqref{eq:proof_estimated_multiplier_bound} gives
\begin{equation}
\sup_{\zeta_{\mathrm{mat}}\in\mathcal Z_{A,n}}
\sup_{\lambda\in\widehat\Lambda_n^{\mathrm{near}}(\tau)}
\left|
\left\langle
\zeta_{\mathrm{mat}}
\{\hat\eta_n(\tau)-\eta^0(P,\tau)\},\lambda
\right\rangle
\right|
\le
C M_n\|\hat\eta_n(\tau)-\eta^0(P,\tau)\|
=o_p(1)
\label{eq:proof_nuisance_substitution_uniform}
\end{equation}
uniformly on \(\mathcal Q_0\). Fixed-row multipliers make no contribution to
the display. The last equality uses
\(M_n(\gamma_n/\sqrt n)^{1/2}=o(1)\).
Equation~\eqref{eq:proof_nuisance_substitution_uniform} is the
required substitution of the estimated minimum-norm completion for its
population counterpart in the bootstrap objective.

\medskip\noindent\emph{Step 3: near-optimal certificates.}
Scale normalization and radial rescaling stabilize the near-optimal dual sets.
Lemma~\ref{lem:directed_near_dual_inclusion} states the required relation
formally. Uniformly over a growing ball of first-order perturbations and a
shrinking neighborhood of the population normalizer, the estimated-row
projection of every perturbed population optimizer lies within \(O_p(a_n)\)
of the sample near-certificate set. Its proof uses
\eqref{eq:proof_estimated_multiplier_bound}, an
\(o(\gamma_n/\sqrt n)\) sample equality residual, and radial scaling to enforce
the sample normalizer and near-objective inequality. Applying the lemma to the
independent-copy, Gaussian, and bootstrap perturbations gives the directed
near-dual-set inclusions used below. The rate \(M_na_n=o(1)\) makes replacing
a perturbed optimizer by its nearby sample certificate negligible in every
first-order stochastic objective.

A population feasible nuisance vector satisfies every fixed row, which is
unchanged in the sample. Thus no positive unbounded fixed-row dual ray occurs
under the null and \(\hat Q_n(\tau)<\infty\) uniformly with probability
approaching one.

\medskip\noindent\emph{Step 4: one-sided comparison and quantile transfer.}

\smallskip\noindent\emph{Independent-copy upper bound.}
Fix \(\theta=(P,\tau)\in\mathcal Q_0\). Let
\((\hat A_n^\dagger,\hat b_n^\dagger(\tau),
\widehat{\mathsf S}_n^\dagger)\) be an independent copy of the complete
sample inputs, and let \(\hat Q_n^\dagger(\tau)\) be its normalized dual
value. It has the same marginal distribution as \(\hat Q_n(\tau)\). Write
\[
(Z_{A,n}^\dagger,Z_{b,n}^\dagger)
:=
\sqrt n\{\hat A_n^\dagger-A_\theta,
\hat b_n^\dagger(\tau)-b_\theta\}.
\]
On the finite-value event, take an optimizer \(\lambda_n^\dagger\) of this
independent-copy program. Population feasibility at
\(\eta^0=\eta^0(P,\tau)\) and
\((\hat A_n^\dagger)^\top\lambda_n^\dagger=0\) give the one-sided saddle
inequality
\begin{equation}
\begin{aligned}
\sqrt n\,\hat Q_n^\dagger(\tau)
&=\sqrt n\,(\hat b_n^\dagger)^\top\lambda_n^\dagger\\
&\le
\left\langle
Z_{b,n}^\dagger-Z_{A,n}^\dagger\eta^0,
\lambda_n^\dagger
\right\rangle .
\end{aligned}
\label{eq:proof_saddle_upper_bound}
\end{equation}
There is no first-order normalizer term in this display: the population
primal relaxation has value and optimal relaxation coordinate equal to zero under the
null. The normalizer still affects \(\lambda_n^\dagger\), which is why
Lemma~\ref{lem:directed_near_dual_inclusion} allows the complete perturbed
normalizer \(\widehat{\mathsf S}_n^\dagger\).

For a conformable perturbation \(z=(z_A,z_b)\), define, conditional on the
original sample,
\[
\mathcal G_{n,\theta}(z)
:=
\sup_{\lambda\in\widehat\Lambda_n^{\mathrm{near}}(\tau)}
\langle z_b-z_A\hat\eta_n(\tau),\lambda\rangle.
\]
Apply Lemma~\ref{lem:directed_near_dual_inclusion} with
\(\zeta=(Z_{A,n}^\dagger,Z_{b,n}^\dagger)\) and
\(\widetilde{\mathsf S}=\widehat{\mathsf S}_n^\dagger\). The resulting
certificate replacement changes the right-hand side of
\eqref{eq:proof_saddle_upper_bound} by at most \(C M_na_n=o(1)\) on the
growing perturbation ball. Equation
\eqref{eq:proof_nuisance_substitution_uniform} then replaces \(\eta^0\) by
\(\hat\eta_n(\tau)\). Let \(F_{n,\theta}^\dagger\) be the event that the
independent-copy root leaves the \(M_n\)-ball, its normalizer leaves the
\(a_n\)-neighborhood or is inadmissible, or
\(\hat Q_n^\dagger(\tau)=+\infty\). Uniform tightness, scale consistency, and
null feasibility give
\[
\delta_n
:=
\sup_{\theta\in\mathcal Q_0}
\PP_P(F_{n,\theta}^\dagger)
\longrightarrow0.
\]
A diagonal-event argument therefore gives a deterministic sequence
\(\beta_n\downarrow0\) and original-sample events with probability approaching
one uniformly on \(\mathcal Q_0\) on which
\begin{equation}
\sqrt n\,\hat Q_n^\dagger(\tau)
\le
\mathcal G_{n,\theta}
(Z_{A,n}^\dagger,Z_{b,n}^\dagger)+\beta_n
+\infty\cdot\mathbbm 1\{F_{n,\theta}^\dagger\}.
\label{eq:proof_one_sided_gaussian_dominance}
\end{equation}

\smallskip\noindent\emph{Gaussian and bootstrap comparison.}
Let \(Z_\theta=(Z_{A,\theta},Z_{b,\theta})\) be a centered Gaussian vector
with the covariance induced by
\(M_\Pi\Sigma_\Pi(P)M_\Pi^\top\), independent of the sample. Let
\(J_{n,\theta}\) denote the CDF of
\(\sqrt n\,\hat Q_n(\tau)\), let \(G_{n,\theta}(\cdot)\) be the conditional
CDF of \(\mathcal G_{n,\theta}(Z_\theta)\), and let
\(H_{n,\theta}(\cdot)\) be the conditional CDF of the bootstrap statistic
\(Q_n^*(\tau)\). The multiplier bound in
\eqref{eq:proof_estimated_multiplier_bound} and the uniform bound on
\(\hat\eta_n(\tau)\) make \(\mathcal G_{n,\theta}\) uniformly Lipschitz in
the estimated perturbation coordinates. Assumption~\ref{ass:sampling},
\eqref{eq:candidate_common_perturbation}, and a Strassen coupling
\citep{strassen1965existence} therefore
imply that there are deterministic \(\varepsilon_n\downarrow0\) and events
\(E_{n,\theta}\) satisfying
\[
\sup_{\theta\in\mathcal Q_0}
\PP_P(E_{n,\theta}^c)\to0
\]
such that, after enlarging \(\varepsilon_n\) to absorb \(\beta_n\) and
\(\delta_n\), on
\(E_{n,\theta}\), for every \(z\in\mathbb R\),
\begin{align}
J_{n,\theta}(z)
&\ge G_{n,\theta}(z-\varepsilon_n)-\varepsilon_n,
\label{eq:proof_cdf_one_sided_transfer}\\
H_{n,\theta}(z)
&\le G_{n,\theta}(z+\varepsilon_n)+\varepsilon_n.
\label{eq:proof_cdf_bootstrap_transfer}
\end{align}
The first inequality uses the independent-copy domination in
\eqref{eq:proof_one_sided_gaussian_dominance}; the second uses the conditional
bootstrap approximation. Candidate indexing adds no
stochastic-equicontinuity term because the primitive input perturbation is
exactly constant in \(\tau\).

Replace \(E_{n,\theta}\) by its intersection with the original-sample
high-probability events established in Steps 1--3. Thus, without changing
\(\sup_{\theta\in\mathcal Q_0}\PP_P(E_{n,\theta}^c)\to0\), assume on this
event that
\(\hat Q_n(\tau)<\infty\),
\(\hat Q_n(\tau)\le\gamma_n/\sqrt n\), the structurally generated sample
system with \(\widetilde\Pi=\hat\Pi_n\), together with
\(\widehat{\mathsf S}_n\), lies in the \(\delta_\Gamma\)-neighborhood in
Assumption~\ref{ass:regular_candidate_systems}(iii), and the multiplier,
directed-inclusion, nuisance-consistency, and scale-consistency bounds above
all hold.

\smallskip\noindent\emph{Positive variance and anti-concentration.}
It remains to transfer the bootstrap quantile through the two CDF
inequalities. Assumption~\ref{ass:sampling} and
\eqref{eq:candidate_common_perturbation} imply that the row-residual Gaussian
perturbation at \(\eta^0(P,\tau)\) has covariance
\[
\Gamma(P,\tau)\Sigma_\Pi(P)\Gamma(P,\tau)^\top.
\]
If \(\Lambda_0(P,\tau)\ne\{0\}\), the first clause of
Assumption~\ref{ass:regular_candidate_systems}(iii) supplies an optimizer
\(\lambda^0\) for which the corresponding quadratic form is at least
\(c_\Gamma\).
Applying Lemma~\ref{lem:directed_near_dual_inclusion} with zero perturbation
places a nearby estimated-row direction in the sample near-certificate set.
Nuisance consistency and bounded covariance preserve the positive variance
bound for the Gaussian objective evaluated at \(\hat\eta_n(\tau)\).

Suppose instead that \(\Lambda_0(P,\tau)=\{0\}\).
Because \(\hat Q_n(\tau)=O_p(n^{-1/2})\) uniformly and
\(\gamma_n\to\infty\), zero belongs to
\(\widehat\Lambda_n^{\mathrm{near}}(\tau)\) with probability approaching one
uniformly. The bootstrap statistic and its critical value are then
nonnegative. On the rejection event, \(\hat Q_n(\tau)>0\). Whenever this score
is finite, its positive optimal face contains a nonzero extreme point
\(\hat\lambda\in\operatorname{ext}(\hat{\mathcal D}_n)\). This optimizer
binds the scale normalization: otherwise it could be radially enlarged, while
a positive fixed-row-only ray would make the score infinite. Hence, on the
uniform scale-consistency event,
\[
1=\mathbbm 1^\top\widehat{\mathsf S}_n\hat\lambda
\le2\bar d\|\hat\lambda_{J_e}\|_1.
\]
Taking \(\widetilde\Pi=\hat\Pi_n\), the sample system and normalizer satisfy
the structural neighborhood in
Assumption~\ref{ass:regular_candidate_systems}(iii), and
\(\hat b_n(\tau)^\top\hat\lambda=\hat Q_n(\tau)>0\). The local
extreme-point clause therefore gives a uniform positive-variance bound for
this rejection-relevant direction. The convergence
\(\hat\eta_n(\tau)\to\eta^0(P,\tau)\), together with bounded system loadings
and bounded covariance matrices, preserves that lower bound when the Gaussian
objective is evaluated at \(\hat\eta_n(\tau)\).

The estimated-row projections of the near-certificate sets are uniformly
bounded, and the Gaussian covariance matrices are uniformly bounded. Choose
\(\bar\alpha\in(\alpha,1/2)\). The Gaussian anti-concentration result in
\citet[Proposition~B.1]{goff2025inference}, together with the positive-variance
direction just established, gives common constants
\(\underline x,C_{\mathrm{ac}}>0\) such that, on the rejection event, the
\((1-\bar\alpha)\)-quantile of \(G_{n,\theta}\) is at least
\(2\underline x\), while its density is bounded by \(C_{\mathrm{ac}}\) on
\([\underline x,\infty)\). For all sufficiently large \(n\),
\(\alpha+\varepsilon_n<\bar\alpha\). From
\eqref{eq:proof_cdf_bootstrap_transfer} and the definition of the bootstrap
critical value,
\[
G_{n,\theta}\{\hat c_n(\tau,1-\alpha)+\varepsilon_n\}
\ge1-\alpha-\varepsilon_n.
\]
Hence
\(\hat c_n(\tau,1-\alpha)+\varepsilon_n\ge2\underline x\), and eventually
\(\hat c_n(\tau,1-\alpha)-\varepsilon_n\ge\underline x\), on
\(E_{n,\theta}\) and the rejection event. This places the relevant interval
away from the possible atom at zero. The density bound then gives, on that
event,
\begin{equation}
G_{n,\theta}
\{\hat c_n(\tau,1-\alpha)+\varepsilon_n\}
-G_{n,\theta}
\{\hat c_n(\tau,1-\alpha)-\varepsilon_n\}
\le 2C_{\mathrm{ac}}\varepsilon_n.
\label{eq:proof_gaussian_anticoncentration}
\end{equation}

\smallskip\noindent\emph{Quantile transfer.}
Write \(c_{n,\theta}:=\hat c_n(\tau,1-\alpha)\) and
\(\bar\varepsilon_n:=(2C_{\mathrm{ac}}+2)\varepsilon_n\). By definition of
the conditional quantile, \(H_{n,\theta}(c_{n,\theta})\ge1-\alpha\).
Combining \eqref{eq:proof_cdf_one_sided_transfer}--
\eqref{eq:proof_gaussian_anticoncentration} gives the explicit chain
\[
\begin{aligned}
J_{n,\theta}(c_{n,\theta})
&\ge G_{n,\theta}(c_{n,\theta}-\varepsilon_n)-\varepsilon_n\\
&\ge G_{n,\theta}(c_{n,\theta}+\varepsilon_n)
-(2C_{\mathrm{ac}}+1)\varepsilon_n\\
&\ge H_{n,\theta}(c_{n,\theta})
-(2C_{\mathrm{ac}}+2)\varepsilon_n\\
&\ge1-\alpha-\bar\varepsilon_n.
\end{aligned}
\]
Let \(q_{n,\theta}(u)\) be the \(u\)-quantile of \(J_{n,\theta}\). On
\(E_{n,\theta}\) and the rejection event, the preceding display implies
\[
q_{n,\theta}(1-\alpha-\bar\varepsilon_n)
\le c_{n,\theta}.
\]
Consequently,
\[
\begin{aligned}
\PP_P\{\sqrt n\,\hat Q_n(\tau)>c_{n,\theta}\}
&\le
\PP_P\!\left[
\sqrt n\,\hat Q_n(\tau)>
q_{n,\theta}(1-\alpha-\bar\varepsilon_n)
\right]
+\PP_P(E_{n,\theta}^c)\\
&\le\alpha+\bar\varepsilon_n+\PP_P(E_{n,\theta}^c).
\end{aligned}
\]
Taking the supremum over \(\mathcal Q_0\) and the limit superior proves the
stated uniform bound and proves the theorem.
\end{proof}

\subsubsection{Power and conservative calibration}

For a pair with \(Q(P,\tau)<\infty\), let
\((r^0(P,\tau),\eta^1(P,\tau))\) solve
\[
\min_{r\ge0,\eta}\ r
\quad\text{subject to}\quad
b(\Pi(P),\tau)-A(\Pi(P))\eta
\le r\,\mathsf S(P)\mathbbm 1,
\]
using the minimum norm \(\eta\) among value-minimizing solutions. Linear
programming duality gives \(r^0(P,\tau)=Q(P,\tau)\).

\begin{proposition}[Power against separated incompatibility]
\label{prop:separated_incompatibility_power}
Fix \(\delta>0\), and let \(\mathcal Q_1(\delta)\) contain pairs
\((P,\tau)\) with \(P\in\mathcal P^U\), \(\tau\in\mathcal T\), and
\[
\delta\le Q(P,\tau)<\infty,
\qquad
\sup_{(P,\tau)\in\mathcal Q_1(\delta)}
\{r^0(P,\tau)+\|\eta^1(P,\tau)\|\}<\infty.
\]
Under Assumption~\ref{ass:sampling} and part (i) of
Assumption~\ref{ass:regular_candidate_systems},
\[
\lim_{n\to\infty}
\sup_{(P,\tau)\in\mathcal Q_1(\delta)}
\PP_P\{\tau\in\mathcal I_n\}=0.
\]
A fixed-row certificate with \(Q(P,\tau)=+\infty\) is rejected without
calibration.
\end{proposition}

\phantomsection
\label{app:separated_power_proof}
\begin{proof}[Proof of Proposition~\ref{prop:separated_incompatibility_power}]
Write the normalized population relaxation in the proposition as a linear
program in \((r,\eta)\). Equation~\eqref{eq:candidate_common_perturbation} and
Lemma~\ref{lem:row_scale_consistency} imply that its sample coefficients,
right-hand side, and row scales converge uniformly to their population values.
On \(\mathcal Q_1(\delta)\), the population minimum norm optimal pair is
uniformly bounded. Substituting this pair into the sample relaxation and
increasing \(r\) by the maximum normalized row perturbation gives
\[
\hat Q_n(\tau)
\le Q(P,\tau)+o_p(1)
\]
uniformly over \(\mathcal Q_1(\delta)\). Fixed rows require no adjustment
because their sample and population entries coincide.

The reverse inequality follows from the same argument applied to a minimum
norm sample optimizer. To see that this optimizer is uniformly bounded, first
use the preceding upper bound and the assumed bound on \(r^0(P,\tau)\) to
bound its relaxation coordinate. The fixed bridge boxes bound \(v^-\) and
\(v^+\). The simplex-vertex comparison rows then bound every coordinate of
\(x\) because the inputs, row scales, and relaxation coordinate are bounded.
Compactness of \(\mathcal T\) also bounds the implied
\(s=\tau\mathbbm 1_{K-1}-x\). Conditional on these quantities, the transformed
admissible lift is a feasible system in \(u\) with fixed coefficients and a
bounded right-hand side. A Hoffman bound therefore gives a uniformly bounded
minimum-norm choice of \(u\). Substituting this bounded sample optimizer into
the population relaxation and increasing its relaxation coordinate by the
maximum normalized entry perturbation gives the reverse value inequality.
Consequently,
\[
\sup_{(P,\tau)\in\mathcal Q_1(\delta)}
|\hat Q_n(\tau)-Q(P,\tau)|=o_p(1).
\]

The near-feasible solution in \eqref{eq:uniform_eta_hat} is bounded by the same
argument. Moreover, part (i) of
Assumption~\ref{ass:regular_candidate_systems} bounds the
total multiplier weight on estimated rows in
\(\widehat\Lambda_n^{\mathrm{near}}(\tau)\). Fixed-row multipliers do not enter
the bootstrap objective because their entry perturbations are zero. The
finite-dimensional bootstrap perturbation is conditionally tight by
Assumption~\ref{ass:sampling} and
\eqref{eq:candidate_common_perturbation}; hence the conditional
critical values are uniformly tight:
\[
\lim_{M\to\infty}\limsup_{n\to\infty}
\sup_{(P,\tau)\in\mathcal Q_1(\delta)}
\PP_P\{\hat c_n(\tau,1-\alpha)>M\}=0.
\]
Uniformly on this class, \(\hat Q_n(\tau)\ge\delta/2\) with probability tending
to one, whereas the critical value remains bounded in probability. Therefore
\(\sqrt n\hat Q_n(\tau)>\hat c_n(\tau,1-\alpha)\) with probability tending to
one uniformly.

Finally, an unbounded normalized dual value can only use an unbounded ray
supported on zero-scale fixed rows. If that ray has a positive objective, it is
unchanged by sampling and gives \(\hat Q_n(\tau)=+\infty\), triggering the
immediate-rejection rule.
\end{proof}

\begin{remark}[Conservative calibration]
\label{rem:conservative_calibration}
Following \citet[discussion preceding Assumption~4.18]{goff2025inference}, one
may reject only when
\[
\sqrt n\,\hat Q_n(\tau)
>
\hat c_n(\tau,1-\alpha)+\varepsilon,
\qquad \varepsilon>0.
\]
For the prespecified bridge system, this shifted rule retains candidatewise
uniform size without part (iii) of
Assumption~\ref{ass:regular_candidate_systems}.
\end{remark}

\end{document}